\documentclass[a4paper]{article}
\usepackage{fullpage}
\usepackage{amssymb,amsmath}
\usepackage{graphicx}
\usepackage{hyperref}

\usepackage{amsfonts}
\usepackage{color}
\newtheorem{df}{Definition}

\newcommand{\qed}{$\Box$}
\newtheorem{theorem}{Theorem}
\newtheorem{lemma}{Lemma}
\newtheorem{cor}{Corollary}
\newenvironment{proof}{\par\noindent{\sf Proof.}}{\hfill\qed\par}

\let\realbfseries=\bfseries
\def\bfseries{\realbfseries\boldmath}
\makeatletter
\let\real@titlestyle=\@titlestyle
\def\@titlestyle{\real@titlestyle\boldmath}
\makeatother

\def\defn#1{\textbf{\textit{\boldmath #1}}}
\let\emph=\defn

\title{\boldmath Any Regular Polyhedron Can Transform to Another by $O(1)$ Refoldings}

\author{
  \begin{tabular}{c@{\qquad}c@{\qquad}c}
  Erik D.~Demaine\thanks{CSAIL, MIT, USA. {\tt \{edemaine,mdemaine,diomidov\}@mit.edu}}
  &
  Martin L.~Demaine\footnotemark[1]
  &
  Yevhenii Diomidov\footnotemark[1]
  \\[\medskipamount]
  Tonan Kamata\thanks{School of Information and Science, Japan
    Advanced Institute of Science and Technology, Japan. {\tt \{kamata,uehara\}@jaist.ac.jp}}
  &
  Ryuhei Uehara\footnotemark[2]
  &
  Hanyu Alice Zhang\thanks{School of Applied and Engineering Physics, Cornell University, USA. {\tt hz496@cornell.edu}}
  \end{tabular}
}
\index{Demaine, Erik D.}
\index{Demaine, Martin L.}
\index{Diomidov, Yevhenii}
\index{Kamata, Tonan}
\index{Uehara, Ryuhei}
\index{Zhang, Hanyu Alice}

\begin{document}
\thispagestyle{empty}
\maketitle

\begin{abstract}
We show that several classes of polyhedra are joined by a sequence of $O(1)$
refolding steps, where each refolding step unfolds the current polyhedron
(allowing cuts anywhere on the surface and allowing overlap)
and folds that unfolding into exactly the next polyhedron;
in other words, a polyhedron is refoldable into another polyhedron
if they share a common unfolding.
Specifically, assuming equal surface area, we prove that
(1)~any two tetramonohedra are refoldable to each other,
(2)~any doubly covered triangle is refoldable to a tetramonohedron,
(3)~any (augmented) regular prismatoid and doubly covered regular polygon is refoldable to a tetramonohedron,
(4)~any tetrahedron has a 3-step refolding sequence to a tetramonohedron, and
(5)~the regular dodecahedron has a 4-step refolding sequence to a tetramonohedron.
In particular, we obtain a $\leq 6$-step refolding sequence between
any pair of Platonic solids, applying (5) for the dodecahedron and
(1) and/or (2) for all other Platonic solids.
As far as the authors know, this is the first result about common unfolding
involving the regular dodecahedron.
\end{abstract}

\section{Introduction}
A polyhedron $Q$ is \emph{refoldable} to a polyhedron $Q'$ if $Q$ can be unfolded to a planar shape that folds into exactly the surface of~$Q'$,
i.e., $Q$ and $Q'$ share a common unfolding/development,
allowing cuts anywhere on the surfaces of $Q$ and $Q'$.
(Although it is probably not necessary for our refoldings, we also allow the common unfolding to self-overlap, as in \cite{NonCommon}.)
The idea of refolding was proposed independently by
M.~Demaine, F.~Hurtado, and E.~Pegg \cite[Open Problem 25.6]{GFA},
who specifically asked whether every regular polyhedron (Platonic solid) can be refolded into any other regular polyhedron.
In this context, there exist some specific results:
Araki et al.~\cite{JZCommon} found two Johnson-Zalgaller solids that are foldable to regular tetrahedra \cite{JZCommon},
and Shirakawa et al.~\cite{SHU2015} found an infinite sequence of polygons that can each fold into a cube and an approaching-regular tetrahedron.

More broadly, Demaine et al.~\cite{Refold} showed that any convex polyhedron can always be refolded to at least one other convex polyhedron.
Xu et al.~\cite{XHSU2017} and Biswas and Demaine~\cite{PrismCommon}
found common unfoldings of more than two (specific) polyhedra.
On the negative side, Horiyama and Uehara~\cite{NonCommon} proved impossibility
of certain refoldings when the common unfolding is restricted to cut along
the edges of polyhedra.

In this paper, we consider the connectivity of polyhedra by the transitive closure of refolding, an idea suggested by Demaine and O'Rourke \cite[Section 25.8.3]{GFA}.
Define a \emph{($k$-step) refolding sequence} from $Q$ to $Q'$ to be
a sequence of convex polyhedra $Q = Q_0, Q_1, \ldots, Q_k = Q'$
where each $Q_{i-1}$ is refoldable to $Q_i$.
We refer to $k$ as the \emph{length} of the refolding sequence.
To avoid confusion, we use ``1-step refoldable'' to refer to the
previous notion of refoldability.

\paragraph{Our results.}

Do all pairs of convex polyhedra of the same surface area
(a trivial necessary condition)
have a finite-step refolding sequence?  If so, how short of a sequence suffices?
As mentioned in \cite[Section 25.8.3]{GFA}, the regular polyhedron open problem
mentioned above is equivalent to asking whether 1-step refolding sequences
exist for all pairs of regular polyhedra.
We solve a closely related problem, replacing ``$1$'' with ``$O(1)$'':
for any pair of regular polyhedra $Q$ and $Q'$,
we give a refolding sequence of length at most~$6$.

More generally, we give a series of results about $O(1)$-step refolding
certain pairs of polyhedra of the same surface area:
\begin{enumerate}
\item In Section~\ref{sec:tetramonohedra},
we show that any two tetramonohedra are 1-step refoldable to each other,
where a \emph{tetramonohedron} is a tetrahedron that consists of four congruent acute triangles.

  This result offers a possible ``canonical form'' for finite-step refolding sequences between any two polyhedra: because a refolding from $Q$ to $Q'$ is also a refolding from $Q'$ to~$Q$, it suffices to show that any polyhedron has a finite-step refolding into some tetramonohedron.
\item In Section~\ref{sec:prismoid}, we show that every regular prismatoid and every augmented regular prismatoid are 1-step refoldable to a tetramonohedron.

  In particular, the regular tetrahedron is a tetramonohedron, the regular hexahedron (cube) is a regular prismatoid,
  and the regular octahedron and regular icosahedron are both augmented regular prismatoids.
  Therefore, the regular tetrahedron has a 2-step refolding sequence to
  the regular hexahedron, octahedron, and icosahedron
  (via an intermediate tetramonohedron); and
  every pair of polyhedra among the regular hexahedron, octahedron,
  and icosahedron have a 3-step refolding sequence
  (via two intermediate tetramonohedra).
\item In Section~\ref{sec:dodecahedron}, we prove that a regular dodecahedron is refoldable to a tetramonohedron by a 4-step refolding sequence.

  As far as the authors know, there are no previous explicit refolding results for the regular dodecahedron, except the general results of \cite{Refold}.

  Combining the results above, any pair of regular polyhedra (Platonic solids)
  have a refolding sequence of length at most 6.
\item In addition, we prove that every doubly covered triangle (Section~\ref{sec:tetramonohedra}) and every doubly covered regular polygon (Section~\ref{sec:prismoid}) are refoldable to a tetramonohedron, and that every tetrahedron has a 3-step refolding sequence to a tetramonohedron (Section~\ref{sec:tetrahedron}).

  Therefore, every pair of polyhedra among the list above have an $O(1)$-step refolding sequence.
\end{enumerate}

\section{Preliminaries}	

For a polyhedron $Q$, $V(Q)$ denotes the set of vertices of $Q$.
For $v \in V(Q)$, define the \emph{cocurvature} $\sigma(v)$ of $v$ on $Q$ to be
the sum of the angles incident to $v$ on the facets of~$Q$.
The \emph{curvature} $\kappa(v)$ of $v$ is defined by
$\kappa(v) = 2\pi-\sigma(v)$.
In particular, if $\kappa(v)=\sigma(v)=\pi$, we call $v$ a \emph{smooth vertex}.
We define \emph{$\Pi_k$} to be the class of polyhedra $Q$ with exactly $k$ smooth vertices.
It is well-known that the total curvature of the vertices of any convex polyhedron is $4\pi$,
by the Gauss--Bonnet Theorem (see \cite[Section 21.3]{GFA}).
Thus the number of smooth vertices of a convex polyhedron is at most 4.
Therefore, the classes $\Pi_0, \Pi_1, \Pi_2, \Pi_3, \Pi_4$ give us a partition of all convex polyhedra.

An \emph{unfolding} of a polyhedron is a (possibly self-overlapping) planar polygon obtained by cutting and developing the surface of the polyhedron (allowing cuts anywhere on the surface).
\emph{Folding} a polygon $P$ is an operation to obtain a polyhedron $Q$ by choosing crease lines on $P$ and gluing the boundary of $P$ properly.
When the polyhedron $Q$ is convex, the following result is crucial:
\begin{lemma}[Alexandrov's Theorem \cite{Alexandrov,GFA}]
  If we fold a polygon $P$ in a way that satisfies the following three \emph{Alexandrov's conditions},
  then there is a unique convex polyhedron $Q$ realized by the folding.
  \begin{enumerate}
  \item Every point on the boundary of $P$ is used in the gluing.
  \item At any glued point, the summation of interior angles (cocurvature) is at most $2\pi$.
  \item The obtained surface is homeomorphic to a sphere.
  \end{enumerate}
\end{lemma}
By this result, when we fold a polygon $P$ to a polyhedron $Q$,
it is enough to check that the gluing satisfies Alexandrov's conditions.
(In this paper, it is easy to check that the conditions are satisfied by our (re)foldings, so we omit their proof.)

A polyhedron $Q$ is \emph{(1-step) refoldable} to a polyhedron $Q'$ if $Q$ can be unfolded to a polygon that folds to $Q'$ (and thus they have the same surface area).
A \emph{($k$-step) refolding sequence} of a polyhedron $Q$ to a polyhedron $Q'$ is a sequence of convex polyhedra $Q=Q_0, Q_1, \ldots, Q_k=Q'$
where $Q_{i-1}$ is refoldable to $Q_i$ for each $i \in \{1, \dots, k\}$.
To simplify some arguments that $Q$ is refoldable to $Q'$,
we sometimes only partially unfold $Q$ (cutting less than needed to make the surface unfold flat), and refold to $Q'$ so that Alexandrov's conditions hold.

We introduce some key polyhedra.
A tetrahedron is a \emph{tetramonohedron} if its faces are four congruent acute triangles.\footnote{This notion is also called 
\emph{isosceles tetrahedron} or \emph{isotetrahedron} in some literature.}
We consider a \emph{doubly covered polygon} as a special polyhedron with two faces.
Precisely, for a given $n$-gon $P$, we make a mirror image $P'$ of $P$ and glue corresponding edges.
Then we obtain a \emph{doubly covered $n$-gon} which has $2$ faces, $n$ edges, and zero volume.


\section{Refoldabilty of Tetramonohedra and Doubly Covered Triangles}
\label{sec:tetramonohedra}

In this section, we first show that any pair of tetramonohedra can be refolded to each other.
We note that a doubly covered rectangle is a (degenerate) tetramonohedron,
by adding edges along two crossing diagonals (one on the front side and
one on the back side).
It is known that a polyhedron is a tetramonohedron if and only if it is in $\Pi_4$ (see, e.g., \cite[p.~97]{AC1979}). 
In other words, $\Pi_4$ is the set of tetramonohedra.
\begin{theorem}
  \label{thm:tetramono}
  For any $Q, Q' \in \Pi_4$, $Q$ is 1-step refoldable to~$Q'$.
\end{theorem}
\begin{proof}
Let $T$ be any triangular face of $Q$. Let $a$ be the length of the longest edge of $T$ and $b$ the height of $T$ for the base edge of length $a$.
We define $T'$, $a'$, and $b'$ in the same manner for $Q'$; refer to \figurename~\ref{fig:Pi_4}.
We assume $a>a'$ without loss of generality.
Now we have $a'>b'$ because $a'$ is the longest edge of $T'$,
and $a'b'=ab$ because $T$ and $T'$ are of the same area.
Thus, $(a')^2 = a' b' \frac{a'}{b'} > a' b' = ab$,
and $2a'>a'>b$ by $a>a'$.

\begin{figure}
  \centering
  \includegraphics[width=8cm]{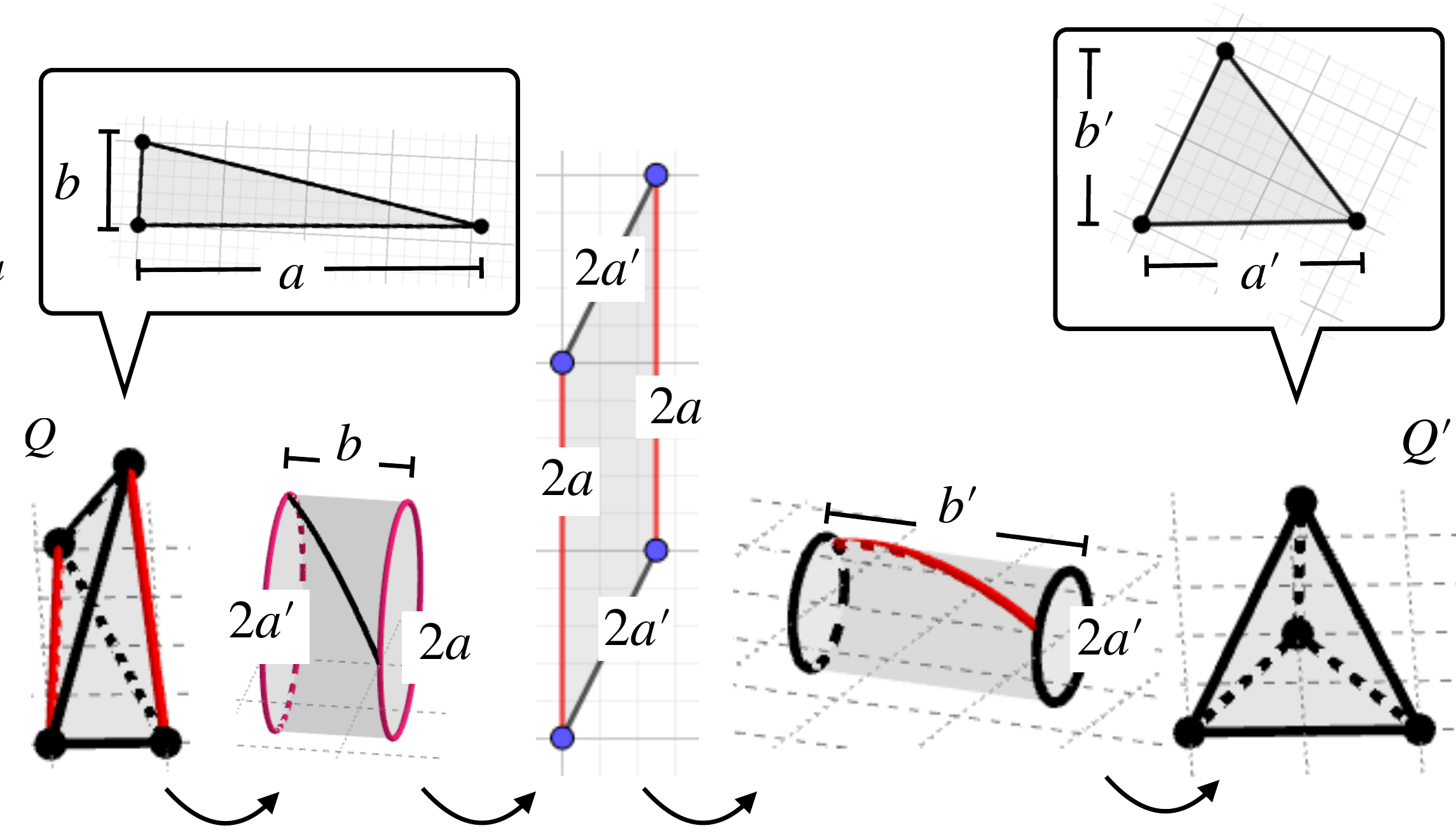}
  \caption{A refolding between two tetramonohedra}
  \label{fig:Pi_4}
\end{figure}

We cut two edges of $Q$ of length $a$, resulting in a cylinder of height $b$ and circumference $2a$.
Then we can cut the cylinder by a segment of length $2a'$ because $2a'>b$.
The resulting polygon is a parallelogram such that two opposite sides have length $2a$ and the other two opposite sides have length $2a'$.
Now we glue the sides of length $2a$ and obtain a cylinder of height $b'$ and circumference $2a'$.
Then we can obtain $Q'$ by folding this cylinder suitably
(the opposite of cutting two edges of $Q'$ of length $2a'$).
%
\end{proof}


To complement the doubly covered rectangles handled by Theorem~\ref{thm:tetramono},
we give a related result for doubly covered triangles:
\begin{theorem}
\label{th:tri}
Any doubly covered triangle $Q$ is 1-step refoldable into a doubly covered rectangle.
Thus, $Q$ has a refolding sequence to any doubly covered triangle~$Q'$ of length at most 3.
If doubly covered triangles $Q$ and $Q'$ share at least one edge length, then
the sequence has length at most 2.
\end{theorem}

\begin{figure}
  \centering
  \includegraphics[width=8cm]{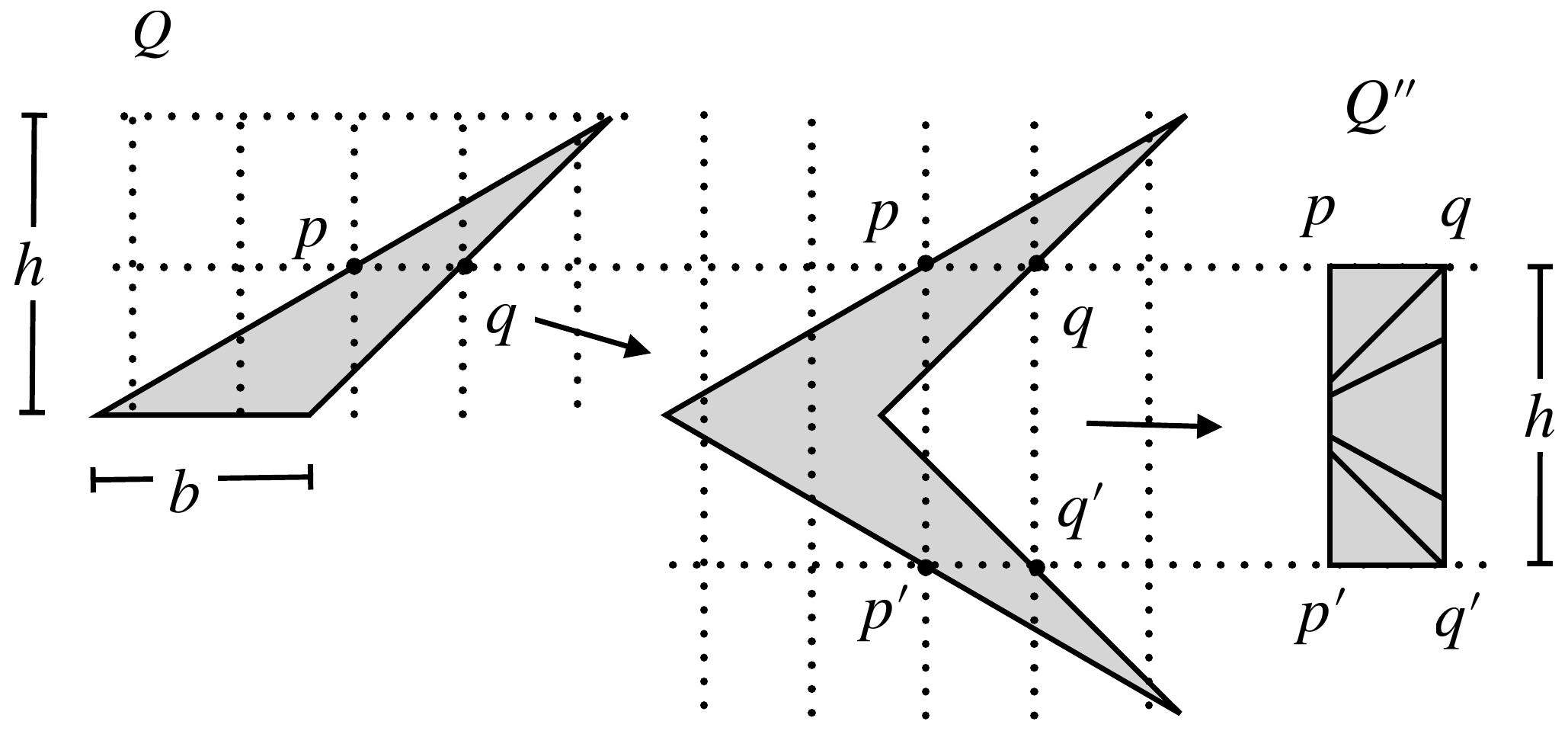}
  \caption{A refolding from a doubly covered triangle to a doubly covered rectangle}
  \label{fig:3to4}
\end{figure}

\begin{proof}
Let $Q$ consist of a triangle $T$ and its mirror image $T'$.
We first cut $Q$ along any two edges, and unfold along the remaining attached edge, resulting in a quadrilateral unfolding as shown in \figurename~\ref{fig:3to4}.
Let $b$ be the length of the uncut edge, which we call the \emph{base}, and let $h$ be the height of $T$ with respect to the base.
Let $p$ and $q$ be the midpoints of the two cut edges.
Then the line segment $pq$ is parallel to the base and of length $b/2$.
In the unfolding of $Q$, let $p'$ and $q'$ be the mirrors of $p$ and $q$, respectively.
Then we can draw a grid based on the rectangle $pp'q'q$ as shown in \figurename~\ref{fig:3to4}.
By folding along the crease lines defined by the grid,
we can obtain a doubly covered rectangle $Q''$ of size $b/2\times h$
(matching the doubled surface area of~$Q$).
(Intuitively, this folding wraps $T$ and $T'$ on the surface of the rectangle $pp'q'q$.)

Because $Q''$ is also a tetramonohedron, the second claim follows from Theorem~\ref{thm:tetramono}.
When $Q$ has an edge of the same length as an edge of $Q'$, as in the third claim, we can cut the other two edges of $Q$ and $Q'$ to obtain the same doubly covered rectangle, resulting in a 2-step refolding sequence.
\end{proof}

The technique in the proof of Theorem~\ref{th:tri} works
for any doubly covered triangle $Q$ even if its faces are acute or
obtuse triangles.

\section{Refoldability of a Regular Prismatoid to a Tetramonohedron}
\label{sec:prismoid}
In this section, we give a 1-step refolding of any regular prism or prismatoid
to a tetramonohedron.
We extend the approach of Horiyama and Uehara \cite{NonCommon}, who showed that the regular icosahedron, the regular octahedron, and the regular hexahedron (cube) can be 1-step refolded into a tetramonohedron.
As an example, \figurename~\ref{fig:commonunfold} shows their common unfolding for the regular icosahedron.

A polygon $P=(p_0, c_1, p_1, c_2, p_2, \ldots, p_{2n}, c_{2n}, p_{2n+1},p_0)$ is called a \emph{spine polygon}
if it satisfies the following two conditions (refer to \figurename~\ref{fig:spine}):
\begin{enumerate}
\item Vertex $p_i$ is on the line segment $p_0 p_n$ for each $0< i< n$;
  vertex $p_i$ is on the line segment $p_{n+1} p_{2n+1}$ for each $n+1 < i < 2n+1$;
  and the polygon $B=(p_0,p_n,p_{n+1},p_{2n+1},p_0)$  is a parallelogram.
  We call $B$ the \emph{base} of $P$, and require it to have positive area.
\item The polygon $T_i=(p_i, c_{i+1}, p_{i+1}, p_i)$ is an isosceles triangle for each $0\le i\le n-1$ and $n+1\le i\le 2n$.
  The triangles $T_0, T_1, \ldots, T_{n-1}$ are congruent, and $T_{n+1}, T_{n+2}, \ldots, T_{2n}$ are also congruent.
  These triangles are called \emph{spikes}.
\end{enumerate}
  
\begin{figure}
  \centering
  \includegraphics[width=8cm]{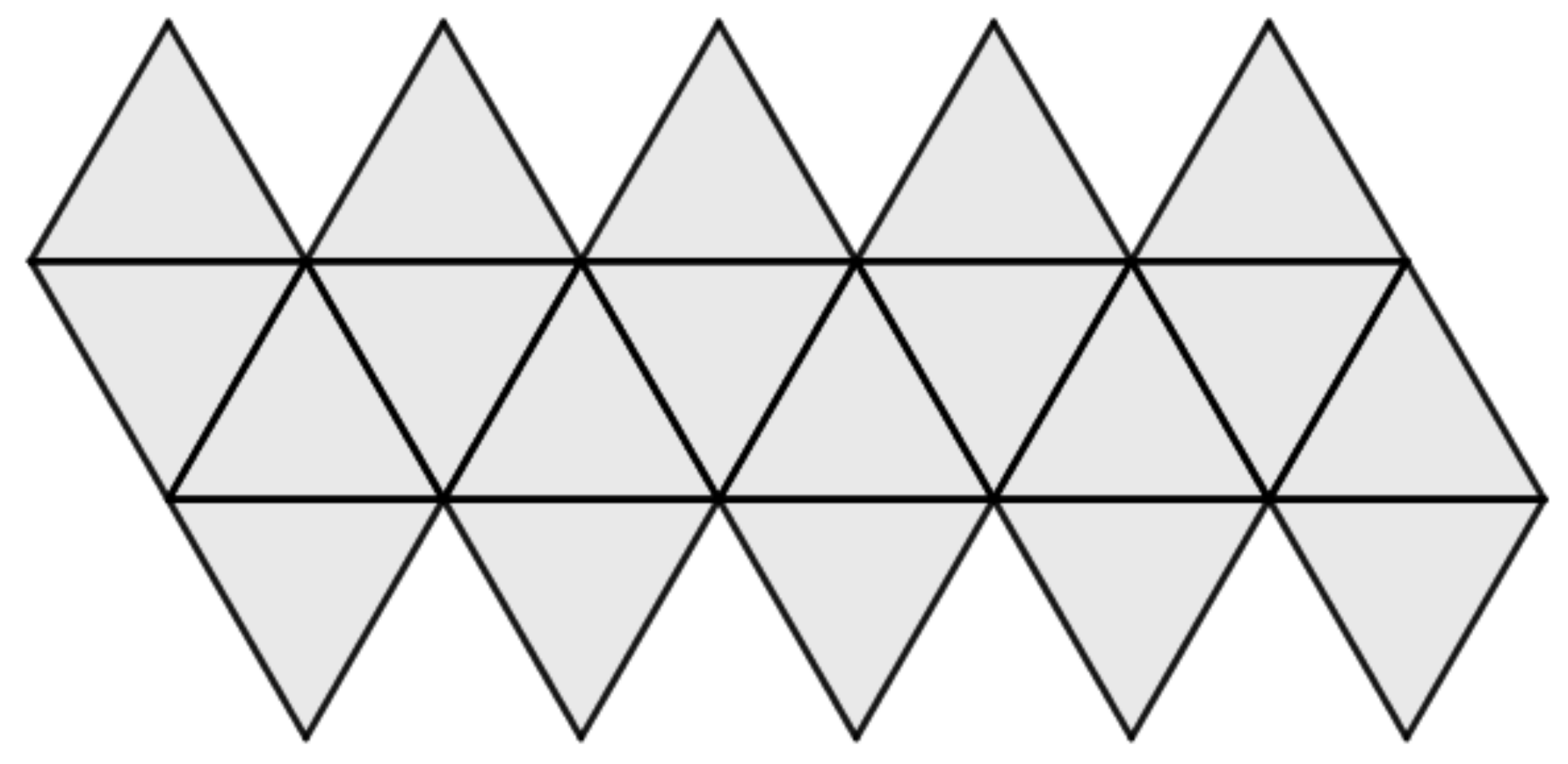}
  \caption{A common unfolding of a regular icosahedron and a tetramonohedron, from \cite{NonCommon}}
  \label{fig:commonunfold}
\end{figure}
\begin{figure}
  \centering
  \includegraphics[width=8cm]{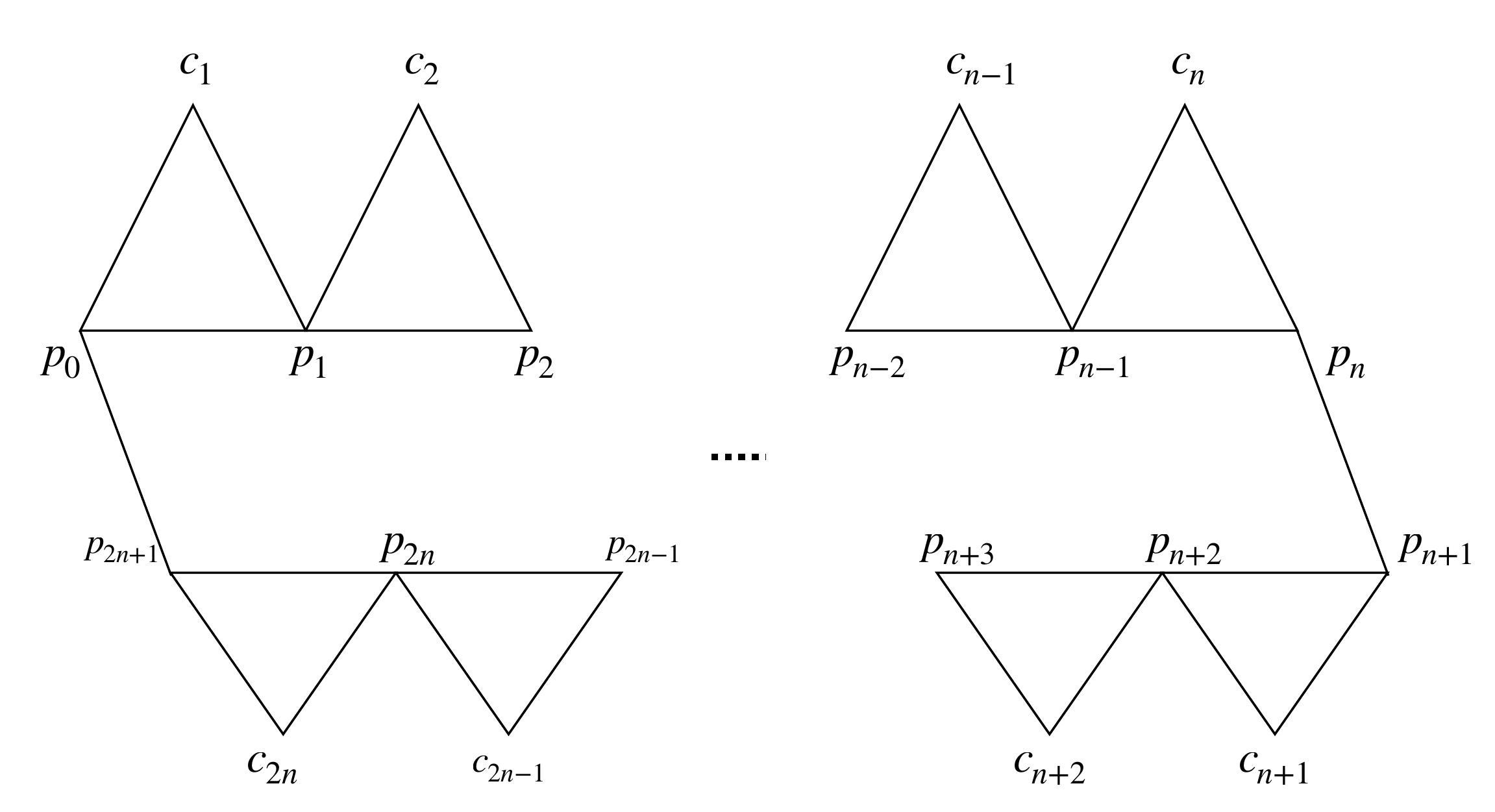}
  \caption{A spine polygon with $2n$ spikes}
  \label{fig:spine}
\end{figure}		

\begin{lemma}
\label{lem:spine}
  Any spine polygon $P$ can be folded to a tetramonohedron.
\end{lemma}
\begin{proof}
  Akiyama and Matsunaga \cite{Isotetra} prove that a polygon $P$ can be folded into a tetramonohedron if the boundary of $P$ can be divided into six parts,
  two of which are parallel and the other four of which are rotationally symmetric.
  We divide the boundary of a spine polygon $P$ into $l_1=(p_0, c_1, \ldots, c_n)$; $l_2=(c_n, p_n)$, $l_3=(p_n, p_{n+1})$; $l_4=(p_{n+1}, c_{n+1}, \ldots,c_{2n})$, $l_5=(c_{2n}, p_{2n+1})$; and $l_6=(p_{2n+1}, p_0)$.
  Then $l_3$ and $l_6$ are parallel because the base of $P$ is a parallelogram.
  Each of $l_2$ and $l_5$ is trivially rotationally symmetric.
  Each of $l_1$ and $l_4$ is rotationally symmetric because each spike of $P$ is an isosceles triangle.
\end{proof}

Now we introduce some classes of polyhedra;
refer to \figurename~\ref{fig:prismatoid}.

A \emph{prismatoid} is the convex hull of parallel \emph{base} and \emph{top} convex polygons. We sometimes call the base and the top \emph{roofs} when they are not distinguished.
We call a prismatoid \emph{regular} if (1) its base $P_1$ and top $P_2$ are congruent regular polygons and (2) the line passing through the centers of $P_1$ and $P_2$ is perpendicular to $P_1$ and $P_2$.
(Note that the side faces of a regular prismatoid do not need to be regular polygons.)
The perpendicular distance between the planes containing $P_1$ and $P_2$ is the \emph{height} of the prismatoid.
The set of regular prismatoids contains \emph{prisms} and \emph{antiprisms},
as well as doubly covered regular polygons (prisms of height zero).

\begin{figure}
  \centering
  \includegraphics[width=8cm]{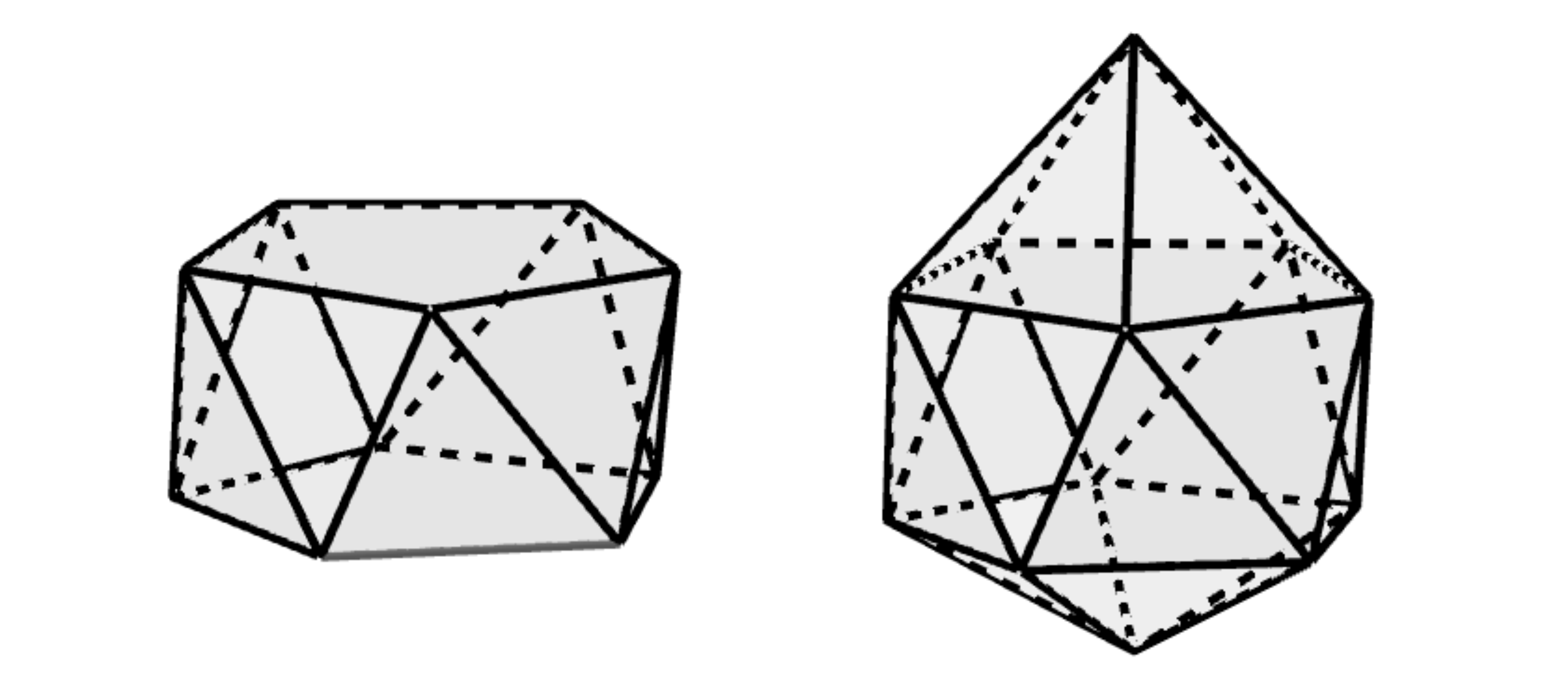}
  \caption{A regular prismatoid and an augmented regular prismatoid}
  \label{fig:prismatoid}
\end{figure}	

A \defn{pyramid} is the convex hull of a \emph{base} convex polygon and an \emph{apex} point.
We call a pyramid \emph{regular} if the base polygon is a regular polygon, and the line passing through the apex and the center of the base is perpendicular to the base.
(Note that the side faces of a regular pyramid do not need to be regular polygons.)
A polyhedron is an \emph{augmented regular prismatoid} if it can be obtained by attaching two regular pyramids 
to a regular prismatoid base-to-roof, where the bases of the pyramids are
congruent to the roofs of the prismatoid and each roof is covered by
the base of one of the pyramids.

\begin{theorem}
  \label{thm:prismoid}
  Any regular prismatoid or augmented regular prismatoid of positive volume can be unfolded to a spine polygon.
\end{theorem}
\begin{proof}
  Let $Q$ be a regular prismatoid. Let $c_1$ and $c_2$ be the center points of two roofs $P_1$ and $P_2$, respectively.
  Cutting from $c_i$ to all vertices of $P_i$ for each $i=1,2$ and
  cutting along a line joining between any pair of vertices of $P_1$ and $P_2$, we obtain a spine polygon.
  For an augmented regular prismatoid $Q$, we can similarly cut from the apex $c_i$ of each pyramid to the other vertices of the pyramid, which are the vertices of the roof $P_i$ of the prismatoid.
\end{proof}

When the height of the regular prismatoid is zero (or it is a doubly covered regular polygon),
the proof of Theorem~\ref{thm:prismoid} does not work because the resulting polygon is not connected.
In this case, we need to add some twist.

\begin{theorem} \label{thm:regular n-gon}
Any doubly covered regular $n$-gon is 1-step refoldable to a tetramonohedron for $n>2$.
\end{theorem}

\begin{proof}
  First suppose that $n$ is an even number $2k$ for some positive integer $k>1$.
  We consider a special spine polygon where the top angles are $\frac{2\pi}{k}$; the vertices $p_0,p_{2n+1},p_1$ are on a circle centered at $c_1$;
  and the vertices $p_{2n+1}, p_1, p_2$ are on a circle centered at $c_{2n}$; see \figurename~\ref{fig:octagon}.
  Then we can obtain a doubly covered $n$-gon by folding along the zig-zag path $p_{2n+1}, p_1, p_{2n}, p_2, \dots, p_{n+2}, p_n$ shown in \figurename~\ref{fig:octagon}.
  Thus when $n=2k$ for some positive integer $k$, we obtain the theorem.

\begin{figure}
  \centering
  \includegraphics[width=8cm]{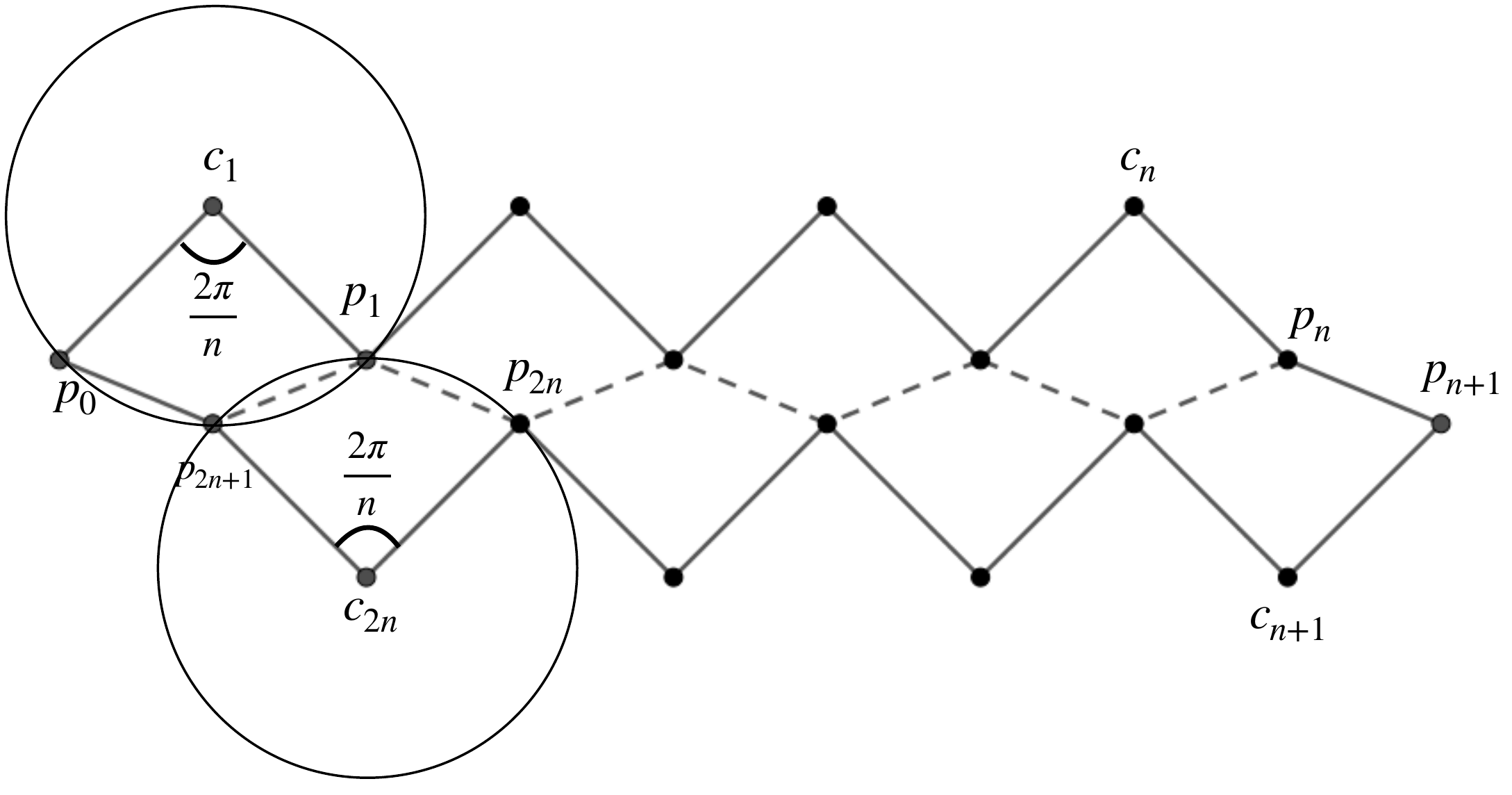}
  \caption{The case of a doubly covered regular $8$-gon}
  \label{fig:octagon}
\end{figure}	

  Now suppose that $n$ is an odd number $2k+1$ for some positive integer $k$.
  We consider the spine polygon whose top angles are $\frac{4\pi}{2k+1}$;
  the vertices $p_0,p_{2n+1},p_1$ are on a circle centered at $c_1$; and the vertices $p_{2n+1}, p_1, p_2$ are on a circle centered at $c_{2n}$.
  From this spine polygon, we cut off two triangles $c_1, p_0, c_{2n+1}$ and $c_{n+1}, p_{n+1}, p_{n}$, as in \figurename~\ref{fig:pentagon}.
  Then we can obtain a doubly covered $n$-gon by folding along the zig-zag path $p_{2n+1}, p_1, p_{2n}, p_2, \dots, p_{n+2}, p_n$ shown in \figurename~\ref{fig:pentagon}.
  Although the unfolding is no longer a spine polygon, it is easy to see that it can also fold into a tetramonohedron
  by letting $l_1'=(c_1, p_1, \ldots, p_n)$, $l_2'=(p_n, p_n)$, $l_3'=(p_n, c_{n+1})$, $l_4'=(c_{n+1}, p_{n+2} \ldots,p_{2n+1})$, $l_5'=(p_{2n+1}, p_{2n+1})$, and $l_6'=(p_{2n+1}, c_1)$
  in the proof of Lemma \ref{lem:spine}.
\end{proof}

\begin{figure}
  \centering
  \includegraphics[width=8cm]{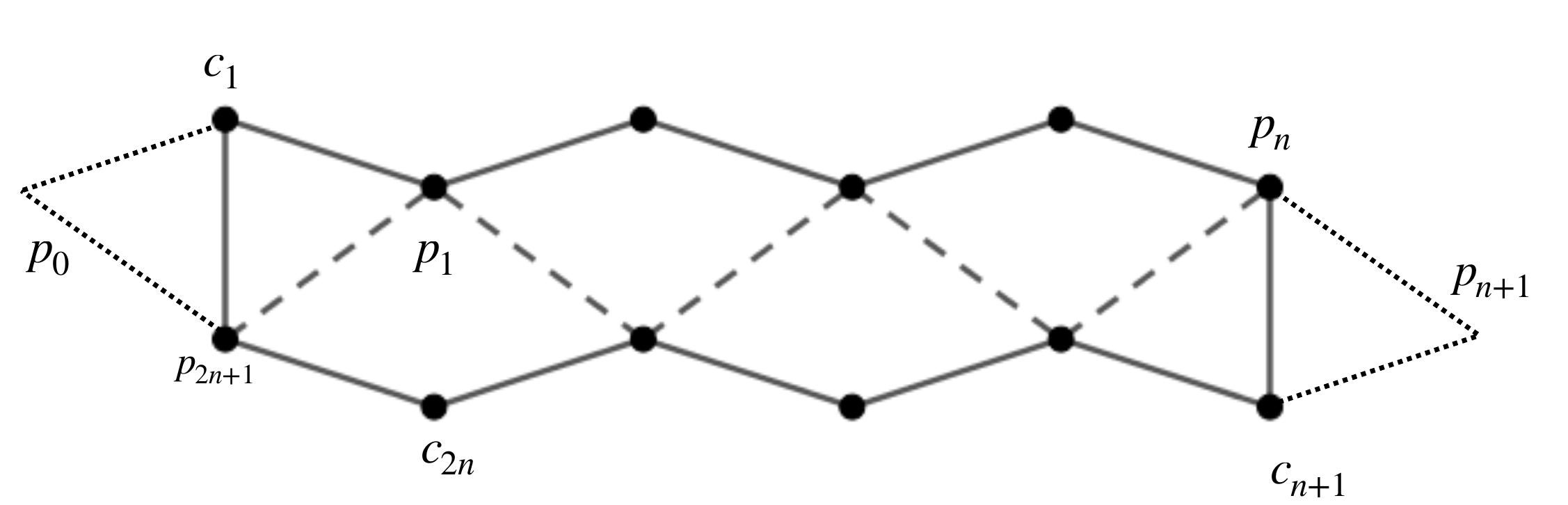}
  \caption{The case of a doubly covered regular $5$-gon}
  \label{fig:pentagon}
\end{figure}	

The proof of Theorem~\ref{thm:regular n-gon} is effectively exploiting that
a doubly covered regular $2k$-gon (with $k>1$) can be viewed as a degenerate regular prismatoid with two $k$-gon roofs,
where each of the side triangles of this prismatoid is on the plane of the roof sharing the base of the triangle.

Because the cube and the regular octahedron are regular prismatoids and the regular icosahedron is an augmented regular prismatoid, we obtain the following:
\begin{cor}
\label{cor:regular}
Let $Q$ and $Q'$ be regular polyhedra of the same area, neither of which is a regular dodecahedron.
Then there exists a refolding sequence of length at most 3 from $Q$ to $Q'$.
When one of $Q$ or $Q'$ is a regular tetrahedron, the length of the sequence is at most 2.
\end{cor}

\begin{figure}[thb]
  \centering
  \includegraphics[width=8cm]{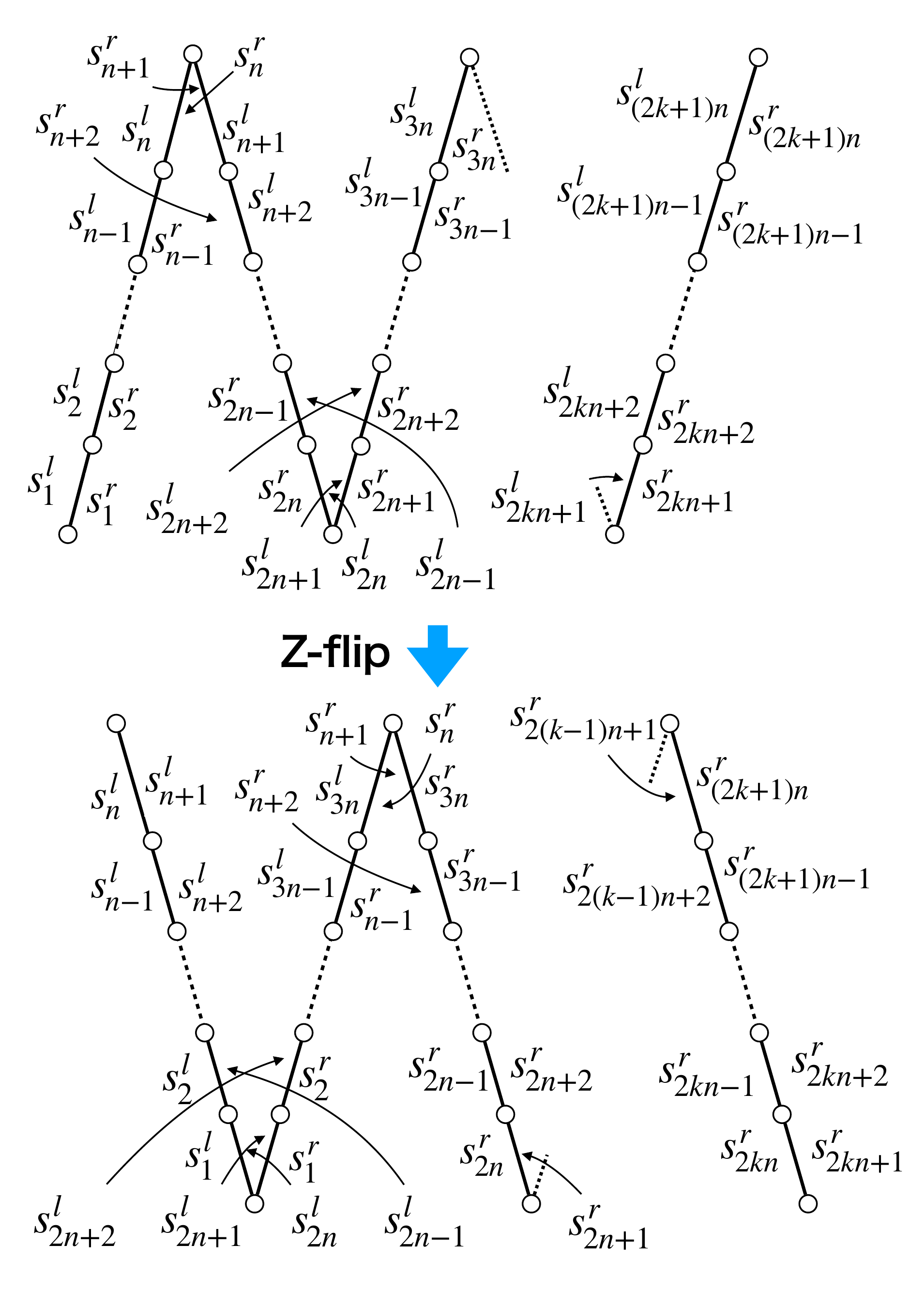}
  \caption{Z-flip}
  \label{fig:Z-flip}
\end{figure}

\section{Refoldability of a Regular Dodecahedron to a Tetramonohedron}
\label{sec:dodecahedron}

In this section, we show that there is a refolding sequence of the regular dodecahedron to a tetramonohedron of length 4.
Combining this result with Corollary~\ref{cor:regular},
we obtain refolding sequences between any two regular polyhedra of length at most 6.

Demaine et al.~\cite{Refold} mention that the regular dodecahedron can be refolded to another convex polyhedron.
Indeed, they show that any convex polyhedron can be refolded to at least one other convex polyhedron
using an idea called ``flipping a Z-shape''. We extend this idea.

\begin{df}\rm
  For a convex polyhedron $Q$ and $n, k\in \mathbb{N}$,
  let $p=(s_1,s_2, \ldots, s_{(2k+1)n})$ be a path that consists of isometric and non-intersecting $(2k+1)n$ straight line segments $s_i$ on $Q$.
  We cut the surface of $Q$ along $p$.
  Then each line segment is divided into two line segments on the boundary of the cut.
  For each line segment $s_i$, let $s^l_i$ and $s^r_i$ correspond to the left and right sides on the boundary along the cut (\figurename~\ref{fig:Z-flip}).
  Then $p$ is a \emph{Z-flippable $(n, k)$-path} on $Q$, and $Q$ is \emph{Z-flippable} by $p$, if the following gluing satisfies Alexandrov's conditions.
  \begin{itemize}
    \item Glue $s_1^l, s_2^l, \ldots, s_n^l$ to $s_{2n}^l, s_{2n-1}^l, \ldots, s_{n+1}^l$.
    \item Glue $s_1^r, s_2^r, \ldots, s_n^r$ to $s_{2n+1}^l, s_{2n+2}^l, \ldots, s_{3n}^l$.
    \item Glue $s_{n+1}^r, s_{n+2}^r, \ldots, s_{2n}^r$ to $s_{3n}^r, s_{3n-1}^r, \ldots, s_{2n+1}^r$.
    $$\vdots$$
    \item \raggedright Glue $s_{2(k-1)n+1}^r, s_{2(k-1)n+2}^r, \ldots, s_{2kn}^r$ to $s_{(2k+1)n}^r, s_{(2k+1)n-1}^r, \ldots, s_{2kn+1}^r$.
  \end{itemize}

  If there are Z-flippable paths $p^1, p^2, \ldots, p^m$ inducing a forest on $Q$, we can flip them all at the same time.
  Then we say that $Q$ is \defn{Z-flippable} by $p^1, p^2, \ldots, p^m$.
\end{df}

\begin{theorem}
  \label{th:dodeca}
  There exists a 4-step refolding sequence between a regular dodecahedron and a tetramonohedron.
\end{theorem}
 
\begin{proof}
  Let $D$ be a regular dodecahedron. To simplify, we assume that each edge of a regular pentagon is of length 1.
  We show that there exists a refolding sequence $D, Q_1, Q_2, Q_3, Q_4$ of length 4 for a tetramonohedron $Q_4$.

  All cocurvatures of the vertices of $D$ are equal to $\frac{9\pi}{5}$.
  For any vertices $v$, there are 3 vertices of distance 1 from $v$
  and 6 vertices of distance $\phi=\frac{1+\sqrt{5}}{2}$ from $v$.
  Hereafter, in figures, each circle describes a non-flat vertex on a polyhedron and
  the number in the circle describes its cocurvature divided by $\frac{\pi}{5}$.
  Each pair of vertices of distance 1 is connected by a solid line,
  and each pair of vertices of distance $\phi$ is connected by a dotted line.
  \figurename~\ref{fig:dodeca} shows the initial state of $D$ in this notation.
  We note that solid and dotted lines do not necessarily imply edges (or crease lines) on the polyhedron.

   \begin{figure}[htb]
    \centering
    \includegraphics[width=7cm]{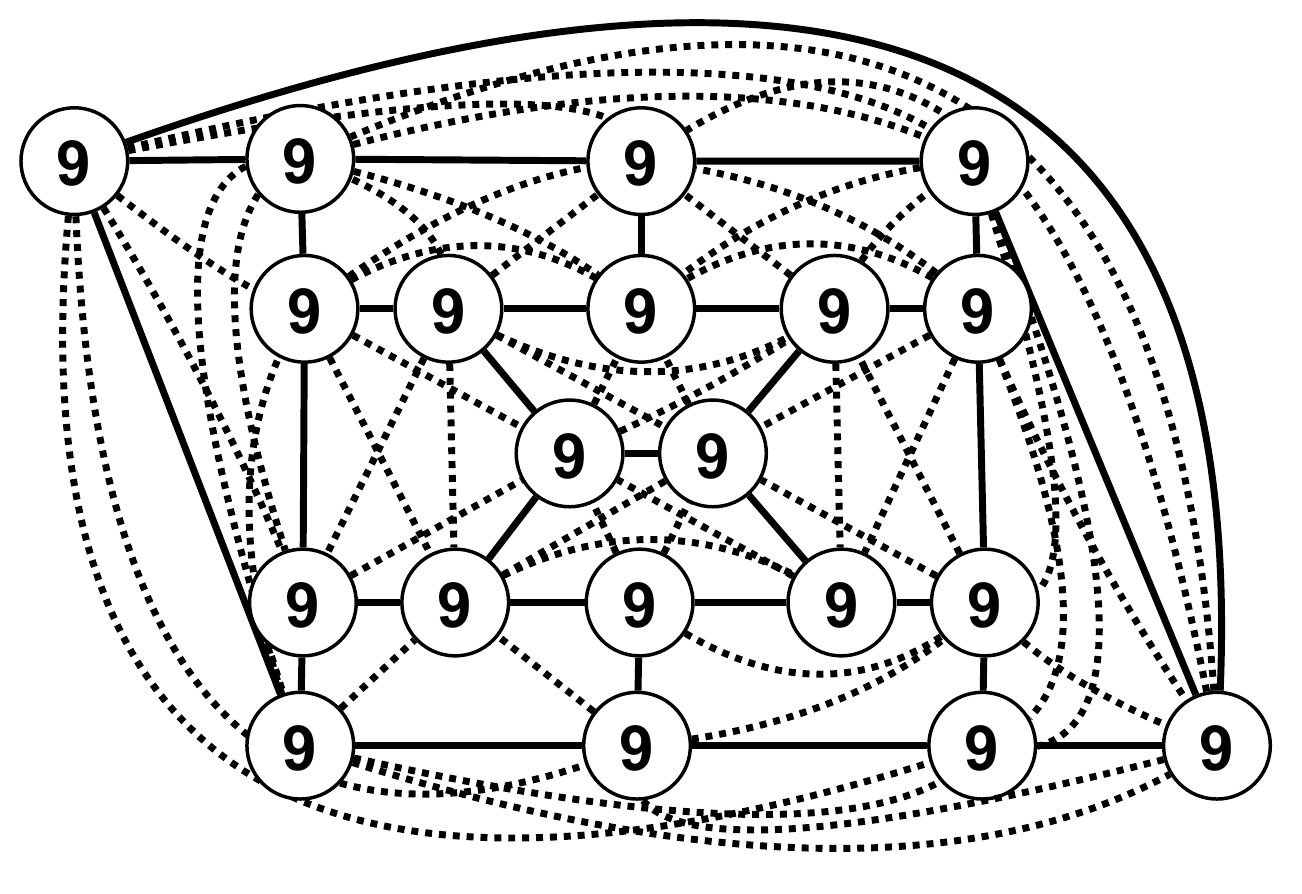}
    \caption{The initial regular dodecahedron}
    \label{fig:dodeca}
  \end{figure}

  \begin{figure}[thb]
    \centering
    \includegraphics[width=0.8\textwidth]{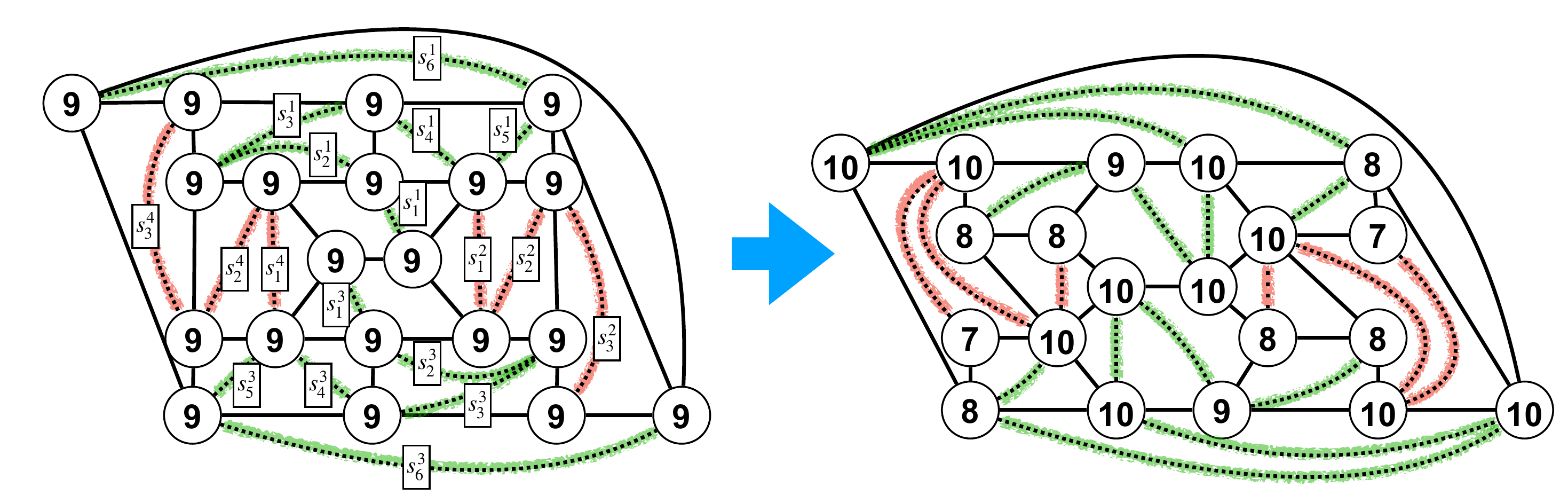}
    \caption{A refolding from $D$ to $Q_1$}
    \label{fig:Step1}
  \end{figure}

 
  First, we choose $p^1=(s^1_1,s^1_2, \ldots, s^1_6)$, $p^2=(s^2_1,s^2_2, s^2_3)$,
  $p^3=(s^3_1,s^3_2, \ldots, s^3_6)$, and $p^4=(s^4_1,s^4_2, s^4_3)$ on the surface of $D$ on the left of \figurename~\ref{fig:Step1}.
  Then, $p^1$ and $p^3$ are Z-flippable $(2, 1)$-paths and $p^2$ and $p^4$ are Z-flippable $(1, 1)$-paths.
  Thus, $D$ is Z-flippable by $p^1, p^2, p^3, p^4$ to the polyhedron on the right of \figurename~\ref{fig:Step1}.
  Let $Q_1$ be the resulting polyhedron. 
  
  \begin{figure}[thb]
    \centering
    \includegraphics[width=0.8\textwidth]{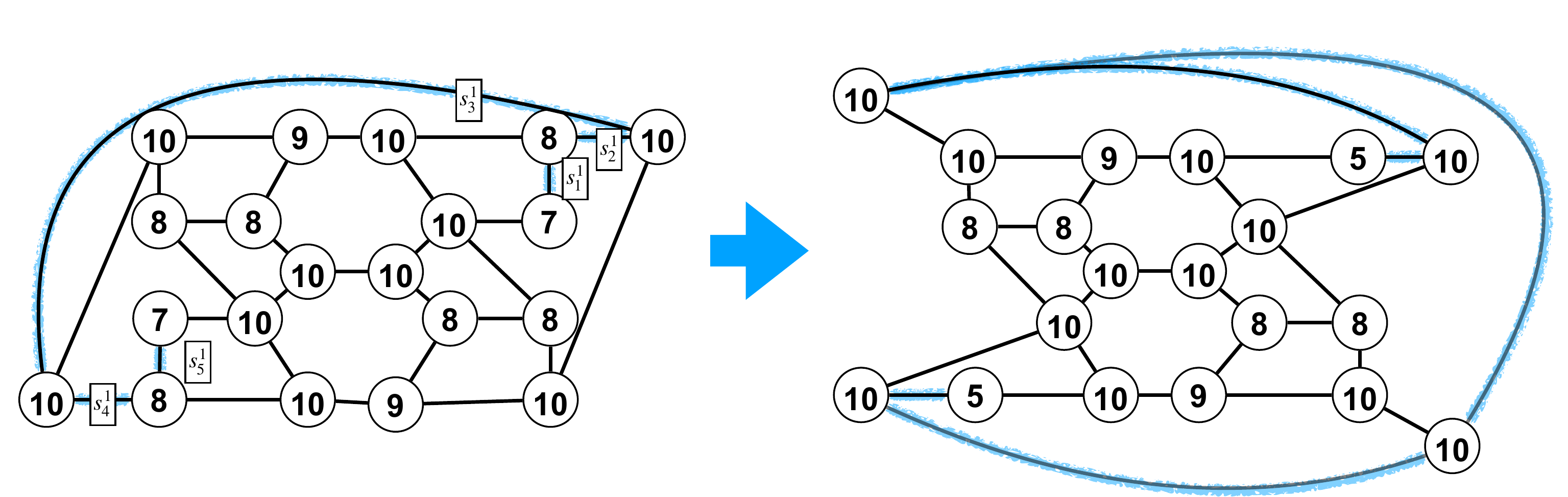}
    \caption{A refolding from $Q_1$ to $Q_2$}
    \label{fig:Step2}
  \end{figure}
  

  Second, we choose $p^1=(s^1_1,s^1_2, \ldots,s^1_5)$ on the surface of $Q_1$ on the left of \figurename~\ref{fig:Step2}.
  Then, $p^1$ is a Z-flippable $(1, 3)$-path. Thus, $Q_1$ is Z-flippable by $p^1$ to the next polyhedron $Q_2$ on the right of \figurename~\ref{fig:Step2}.
  
  \begin{figure}[thb]
    \centering
    \includegraphics[width=0.8\textwidth]{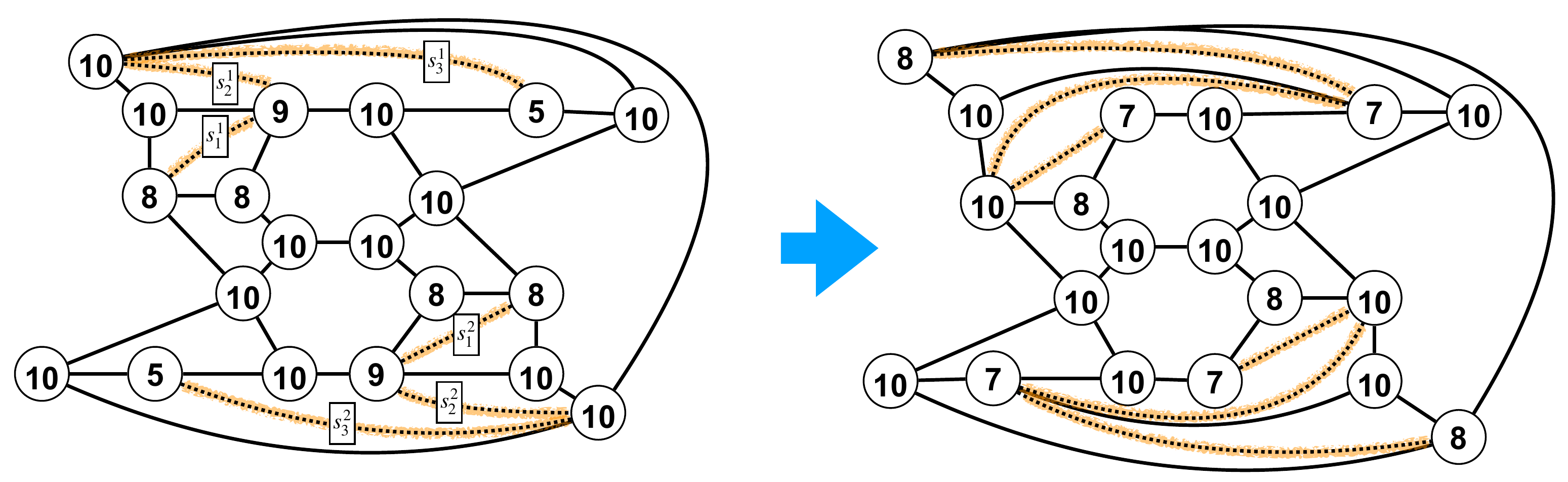}
    \caption{A refolding from $Q_2$ to $Q_3$}
    \label{fig:Step3}
  \end{figure}


  Third, we choose $p^1=(s^1_1,s^1_2,s^1_3)$ and $p^2=(s^2_1,s^2_2,s^2_3)$ on the surface of $Q_2$ on the left of \figurename~\ref{fig:Step3}.
  Then, $p^1$ and $p^2$ are Z-flippable $(1, 1)$-paths. Thus, $Q_2$ is Z-flippable by $p^1$ and $p^2$ to the polyhedron $Q_3$ on the right of \figurename~\ref{fig:Step3}.

  \begin{figure}[thb]
    \centering
    \includegraphics[width=0.8\textwidth]{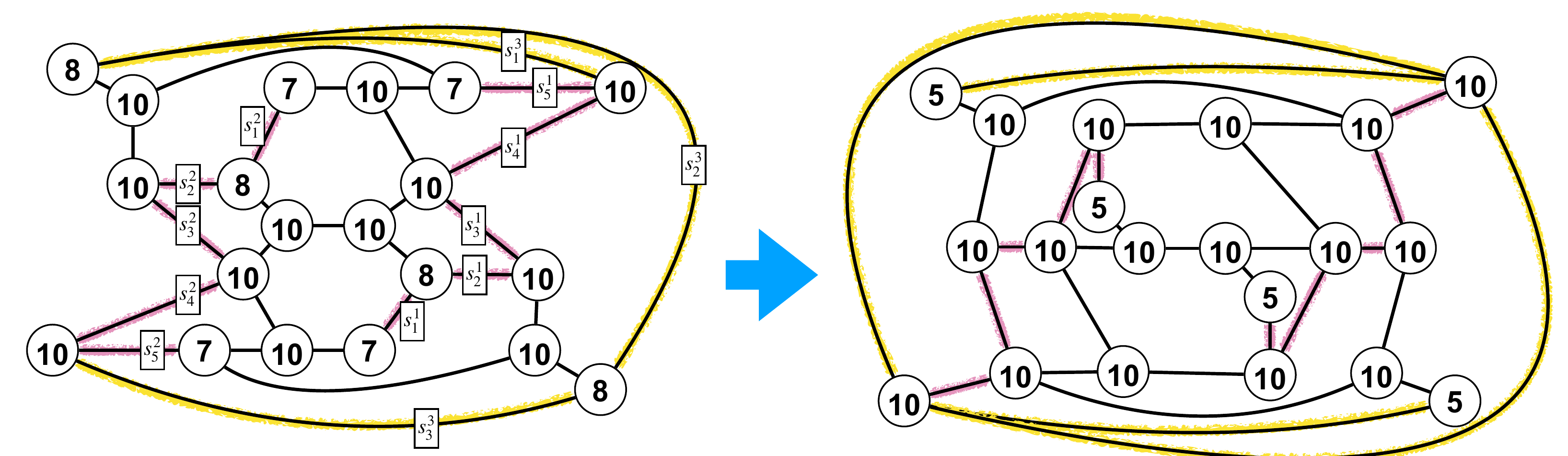}
    \caption{A refolding from $Q_3$ to $Q_4$}
    \label{fig:Step4}
  \end{figure}


  Fourth, we choose $p^1=(s^1_1,s^1_2, \ldots, s^1_5)$, $p^2=(s^2_1,s^2_2,\ldots,s^2_5)$, and $p^3=(s^3_1,s^3_2,s^3_3)$
  on the surface of $Q_3$ on the left of \figurename~\ref{fig:Step4}.
  Then, $p^1$ and $p^2$ are Z-flippable $(1, 3)$-paths and $p^3$ is a Z-flippable $(1, 1)$-path.
  Thus, $Q_3$ is Z-flippable by $p^1$, $p^2$, and $p^3$ to the polyhedron $Q_4$ on the right of \figurename~\ref{fig:Step4}.
  Finally, we obtain a tetramonohedron $Q_4$ from a regular dodecahedron $D$ by a 4-step refolding sequence.

In this proof, we used partial unfolding between pairs of polyhedra in the refolding sequence.
We give the (fully unfolded) common unfoldings in Appendix~\ref{appendix}.
Thus, there exists a 4-step refolding sequence between a regular dodecahedron and a tetramonohedron.
\end{proof}	

\section{Refoldability of a Tetrahedron to a Tetramonohedron}
\label{sec:tetrahedron}
In this section, we prove that any tetrahedron can be refolded to a tetramonohedron.
Let $\mathcal{Q}_k$ denote the class of polyhedra with exactly $k$ vertices.

\subsection{Refoldability of $\Pi_3$ to $\Pi_4$}
First we show a technical lemma: any polyhedron $Q$ in $\Pi_3$ can be refolded to a tetramonohedron by
a refolding sequence of length linear in the number of vertices of~$Q$.
\begin{lemma}
  \label{lem:Pi_3}
  For any $Q\in \Pi_3\cap \mathcal{Q}_n$ with $n \geq 5$,
  there is a refolding sequence of length
  $2n-9$ from $Q$ to some $Q' \in \Pi_4$.
\end{lemma}	
\begin{proof}(Outline)
We prove the claim by induction.
As the base case, suppose $n=5$; refer to \figurename~\ref{fig:Pi_3toPi_4-1}.
Let $\lambda_1,\lambda_2,\lambda_3$ be the smooth vertices of $Q$ and $v_i,v_j$ be the other vertices.
Take a point $m$ on the segment $\lambda_1 \lambda_2$ with $\angle v_iv_jm = \kappa(v_i)$ and
cut the surface of $Q$ along the segments $\lambda_1 \lambda_2$, $v_i v_j$, and $v_j m$.
Then $\sigma (v_j) - \kappa (v_i) = \pi$ because $\kappa (v_i)+\kappa (v_j) = \pi$.
The point $v_j$ is then divided into a point of degree $\kappa(v_i)$ and a point of degree $\pi$ on the boundary.
Trace the obtained boundary from $v_i$ counterclockwise and denote points corresponding to $v_i, v_j, m, m$ by $p_i, p_{j}, p_{m_1}, p_{m_2}$, respectively.
Let $c_1$ and $c_2$ be the center points of the segments $p_i p_{j}$ and $p_{m_1} p_{m_2}$, respectively.
We take the point $s$ which has the same distance with $p_{m_1}$ from $c_1$.
Let $c_3$ be the center of the segment $p_{m_2} s$.
Now glue the segment $s c_1$ to $c_1 p_{m_1}$, the segment $p_{m_2} c_2$ to $c_2 p_{m_1}$, and the segment $p_{m_2} c_3$ to $c_3,s$.
Let $Q'$ be the obtained polyhedron after the gluing.
Then $Q'$ is in $\Pi_4$ because each curvature of every vertex of $Q'$ is $\pi$.

\begin{figure}
  \centering
  \includegraphics[width=7cm]{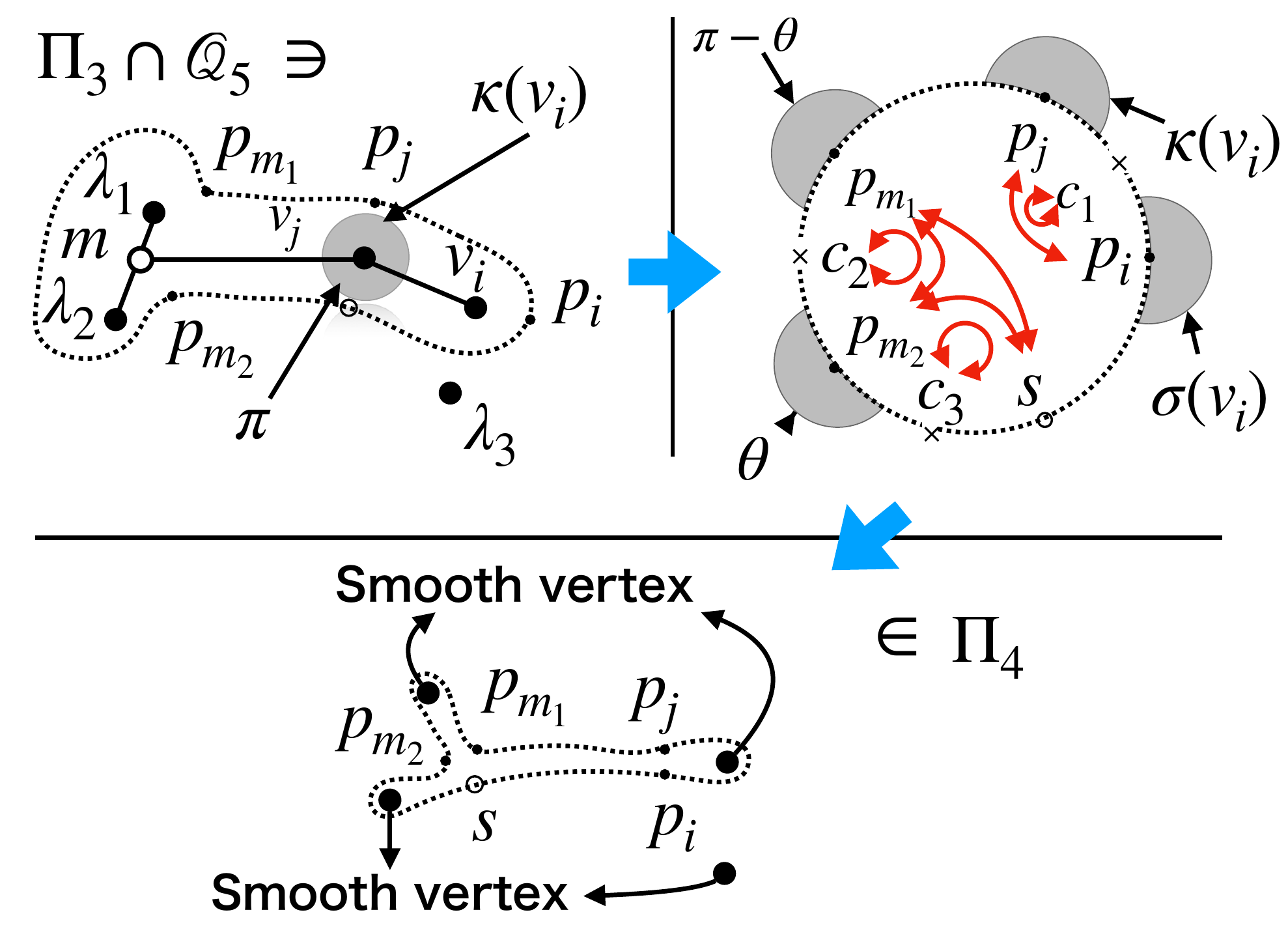}
  \caption{A refolding of a polyhedron in $\Pi_3\cap \mathcal{Q}_5$ to $\Pi_4$}
  \label{fig:Pi_3toPi_4-1}
\end{figure}

Now we turn to the inductive step.
Let $Q$ be any polyhedron in $\Pi_3\cap \mathcal{Q}_{k}$ with $k>5$.
We prove that there exists a refolding sequence $(Q,Q',Q'')$
for two polyhedra $Q'$ and $Q''$ with $Q''\in \Pi_3\cap \mathcal{Q}_{k-1}$;
refer to \figurename~\ref{fig:Pi_3toPi_4-2}.

\begin{figure}
  \centering
  \includegraphics[width=7cm]{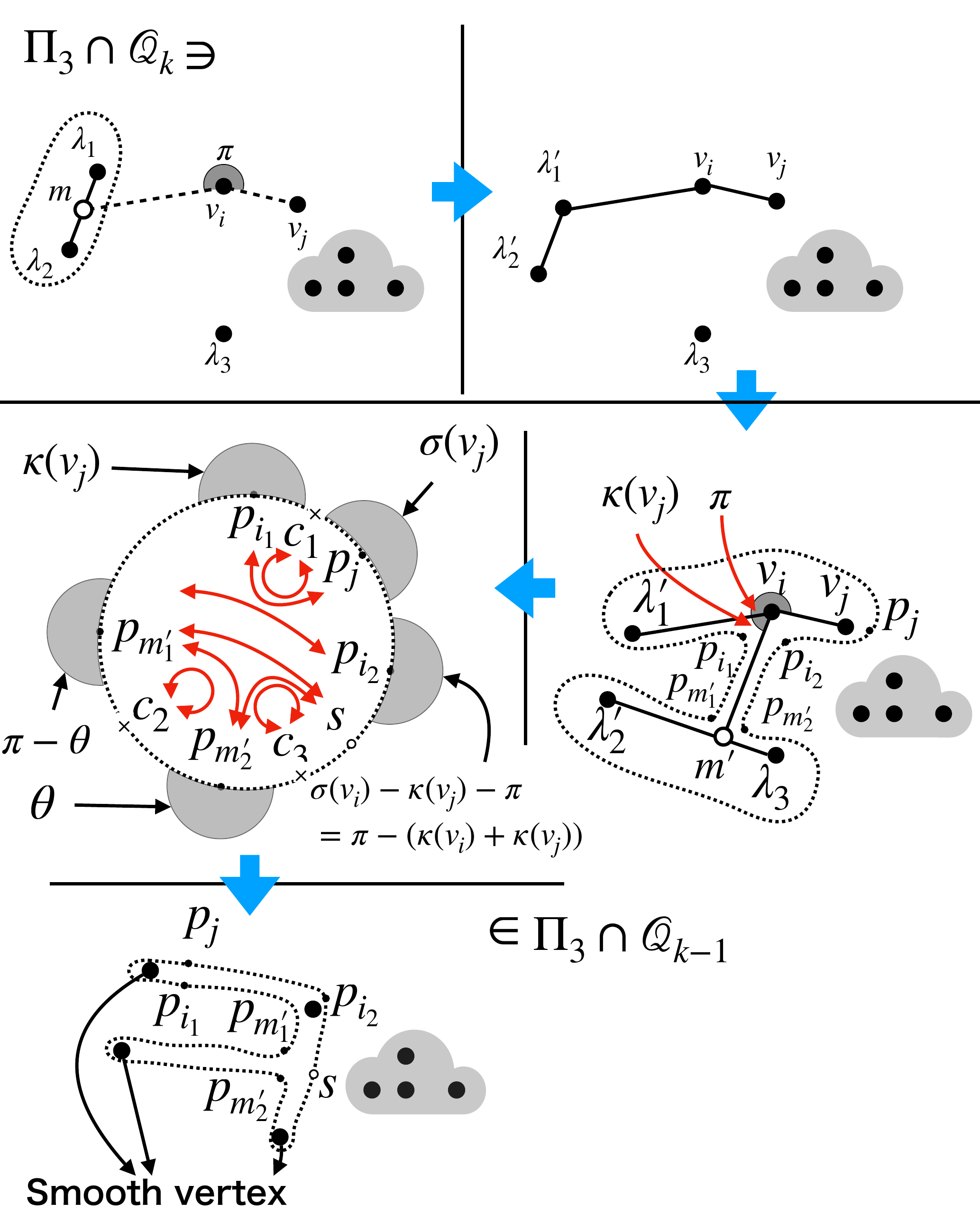}
  \caption{A refolding sequence from $\Pi_3\cap \mathcal{Q}_k$ to $\Pi_3\cap \mathcal{Q}_{k-1}$}
  \label{fig:Pi_3toPi_4-2}
\end{figure}

Let $\lambda_1,\lambda_2,\lambda_3$ be the smooth vertices of $Q$ and $v_i,v_j$ be any other vertices on $Q$.
Note that $0 < \kappa (v_i) + \kappa (v_j) < \pi$ because $k>5$.
Take a point $m$ on the segment $\lambda_1 \lambda_2$ with $\angle v_iv_jm = \pi$, cut the surface of $Q'$ along the segment $\lambda_1 \lambda_2$,
and glue it again so that $m$ is an endpoint. (This can be done because the cut produces a ``rolling belt'' in terms of folding; see \cite{GFA} for the details.)
Let $Q'$ be the obtained polyhedron, and let $\lambda'_1(=m),\lambda'_2,\lambda'_3(=\lambda_3)$ be the smooth vertices of $Q'$.
Now take a point $m'$ on the segment $\lambda'_2,\lambda'_3$ such that $\angle \lambda'_1 v_i m' = \kappa(v_j)$ and cut the surface of $Q'$ along
the segments $\lambda'_1 v_i$, $\lambda'_2 \lambda'_3$, $v_i v_j$, and $v_i m'$.
Trace the obtained boundary from $v_j$ counterclockwise and denote points corresponding to $v_j, v_i, m', m', v_i$ by $p_j, p_{i_1}, p_{m'_1}, p_{m'_2}, p_{i_2}$, respectively.
Let $c_1$ and $c_2$ be the midpoints of $p_{i_1} p_j$ and $p_{m'_1} p_{m'_2}$, respectively.
Take the point $s$ which has the same distance with $p_{m'_1}$ from $c_1$.
Let $c_3$ be the midpoint of the segment $p_{m'_2} s$.
Now glue the segment $p_{m'_1} c_1$ to $c_1 s$, the segment $p_{m'_2} c_2$ to $c_2 p_{m'_1}$, and the segment  $p_{m'_2} c_3$ to $c_3 s$.
Let $Q''$ be the obtained polyhedron.
Then $Q''$ is in $\Pi_3 \cap \mathcal{Q}_{k-1}$ because each curvature of points corresponding to $c_1,c_2,c_3$ is $\pi$ and
$0 < \sigma (v'_4) < \pi$ because $\sigma (v'_4) = \sigma (v_i) - \kappa(v_j) = \sigma (v_i) + \sigma(v_j) - 2\pi$ and $0 < \kappa (v_i) + \kappa (v_j) < \pi$.
\end{proof}

\begin{figure}
  \centering
  \includegraphics[width=7cm]{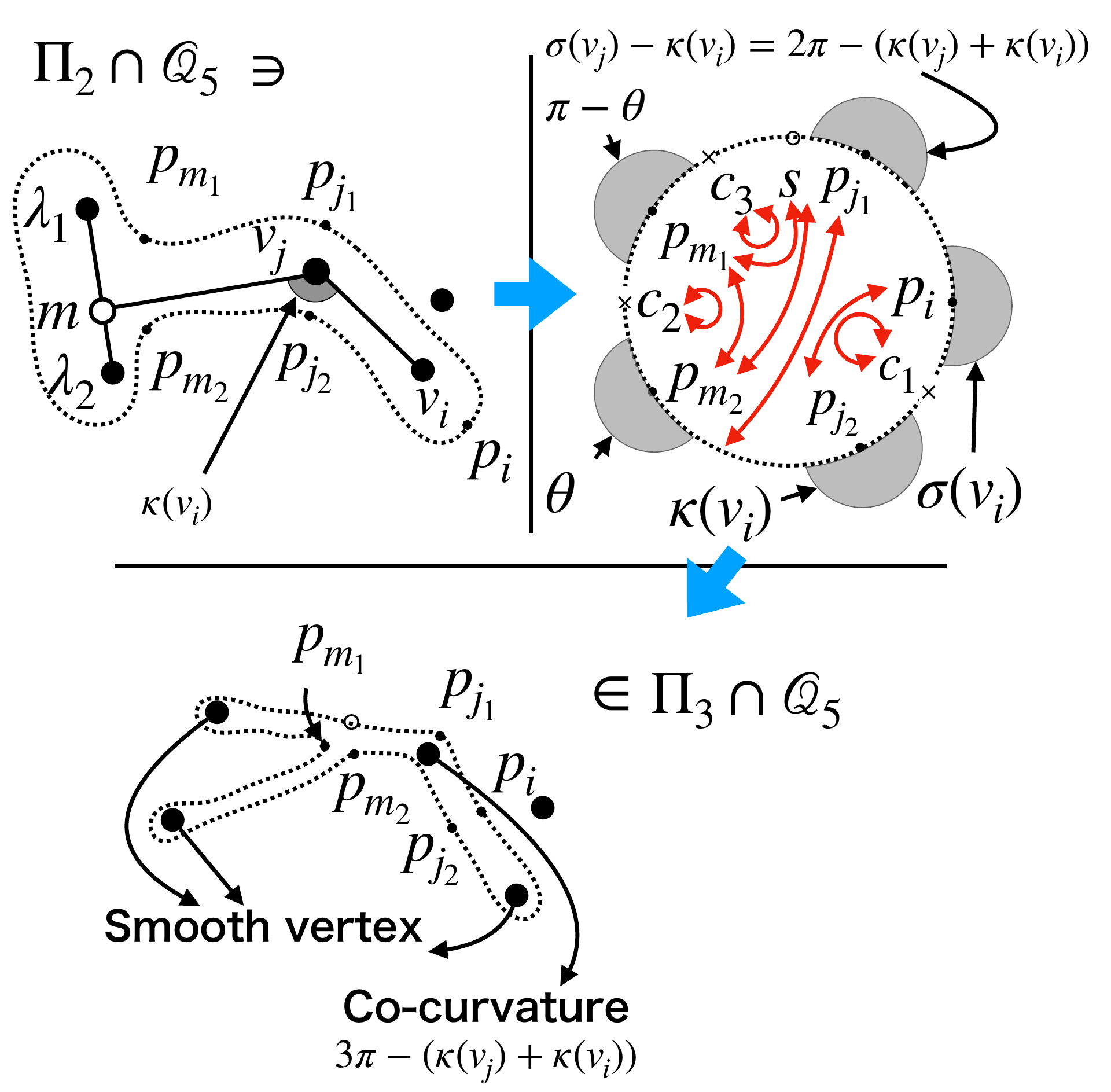}
  \caption{Refolding from any tetrahedron to a tetramonohedron}
  \label{fig:tetra}
\end{figure}

\begin{theorem}
  \label{thm:tetra}
  For any $Q \in \mathcal{Q}_4$, there is a 3-step refolding sequence from $Q$ to some $Q''' \in \Pi_4$.
\end{theorem}

\begin{proof}(Outline)
  Let $v,v'$ be two vertices of $Q$ with smallest cocurvature.
  (That is, $\sigma(v),\sigma(v')\le \sigma(v'')$ for the other two vertices $v''$ of $Q$.)
  We cut along the segment $v v'$ and glue the point $v$ to $v'$. Let $Q'$ be the resulting polyhedron.
  Then, because $\sigma(v) + \sigma(v') \leq 2\pi$ by the Gauss--Bonnet Theorem, $Q'$ satisfies the Alexandrov's conditions.
  That is, $Q'$ is a convex polyhedron in $\Pi_2 \cup \mathcal{Q}_5$.
  (We assume that the original $Q$ has no smooth vertex to simplify the arguments.)

  Let $\lambda_1$ and $\lambda_2$ be the two smooth vertices of $Q'$ (which was generated by the gluing of $v$ and $v'$),
  and $v_i,v_j$ be two vertices of $Q'$ of larger cocurvature than others with $\kappa(v_i)<\kappa(v_j)$.
  By the Gauss--Bonnet Theorem, $\frac{4\pi}{3} \leq \sigma(v_i) + \sigma(v_j) \leq 2\pi$.
  We take the point $m$ on the segment $\lambda_1 \lambda_2$ so that $\angle(v_i,v_j,m) = \kappa(v_i)$.
  Now we cut along the segments $\lambda_1 \lambda_2$, $v_i v_j$, and $v_j m$; see \figurename~\ref{fig:tetra}.
  Trace the obtained boundary from $v_i$ counterclockwise and denote points corresponding to $v_i, v_j, m, m, v_j$ by
  $p_i, p_{j_1}, p_{m_1}, p_{j_2}, p_{m_2}$, respectively.
  Let $c_1$ and $c_2$ be the midpoints of the segments $p_i p_{j_2}$ and $p_{m_1} p_{m_2}$, respectively.
  Furthermore, we take the point $s$ which has the same distance with $p_{m_2}$ from $c_1$,
  and let $c_3$ be the midpoint of $s p_{m_1}$.

  Now glue segment $p_i c_1$ to $p_{j_2} c_1$,
  segment $p_i s$ to $p_{j_2} p_{m_2}$,
  segment $p_{m_1} c_2$ to $p_{m_2} c_2$, and
  segment $s c_3$ to $p_{m_1} c_3$.
  Let $Q''$ be the resulting polyhedron. Then the gluing to fold $Q''$ produces four vertices.
  Among them, three vertices produced by the points $c_1, c_2, c_3$ are smooth vertices of curvature $\pi$.
  The cocurvature of the vertex of $Q''$ generated by the gluing of $p_{j_1}$ to the boundary is
  $3\pi-(\kappa(v_j)+\kappa(v_i))$, which is in $[\pi, \frac{5\pi}{3}]$ by
  $\frac{4\pi}{3} \leq \sigma(v_i) + \sigma(v_j) \leq 2\pi$.
  Therefore, $Q''$ satisfies Alexandrov's conditions, and hence we obtain $Q'' \in \Pi_3 \cup \mathcal{Q}$.
  By Lemma~\ref{lem:Pi_3}, there exists $Q''' \in \Pi_4$ such that there is a 3-step refolding sequence from $Q$ to $Q'''$.
\end{proof}

\section{Conclusion}

In this paper, we give a partial answer to Open Problem 25.6 in \cite{GFA}.
For every pair of regular polyhedra, we obtain a refolding sequence of length at most 6.
Although this is the first refolding result for the regular dodecahedron, the number of refolding steps to other regular polyhedra seems a bit large.
Finding a shorter refolding sequence than Theorem~\ref{th:dodeca} is an open problem.

The notion of refolding sequence raises many open problems.
What pairs of convex polyhedra are connected by a refolding sequence of finite length?
Is there any pair of convex polyhedra that are not connected by any refolding sequence?

At the center of our results is that the set of tetramonohedra
induces a clique by the binary relation of refoldability.
Is the regular dodecahedron refoldable to a tetramonohedron?
Are all Archimedean and Johnson solids refoldable to tetramonohedra?
Is there any convex polyhedron not refoldable to a tetramonohedron?
(If not, we would obtain a 3-step refolding sequence between any pair of
convex polyhedra.)

Another open problem is the extent to which allowing or forbidding overlap
in the common unfoldings affects refoldability.
While we have defined refoldability to allow overlap,
in particular to follow \cite{Refold} where it may be necessary,
most of the results in this paper would still apply if we forbade overlap.
For example, Appendix~\ref{appendix} confirms this for our refolding sequence
from the regular dodecahedron to a tetramonohedron;
while the general approach of Lemma~\ref{lem:Pi_3}
is likely harder to generalize.
Are there two polyhedra that have a common unfolding
but all such common unfoldings overlap?
(If not, the two notions of refolding are equivalent.)

\section*{Acknowledgments}

This work was initiated during MIT class 6.849: Geometric Folding Algorithms,
Fall 2020.

Erik Demaine was partially supported by the Cornell Center for Materials
Research with funding from the NSF MRSEC program (DMR-1120296).
Hanyu Zhang was primarily supported by the Cornell Center for Materials Research
with funding from the NSF MRSEC program (DMR-1719875).
Ryuhei Uehara was partially supported by MEXT/JSPS Kakenhi Grant JP17H06287 and JP18H04091.



\nocite{U2020}
\small
\bibliographystyle{plain}
\bibliography{main}

\appendix
\section{Common Unfoldings from Regular Dodecahedron to Tetramonohedron in Theorem \ref{th:dodeca}}
\label{appendix}

Figures~\ref{fig:Q0toQ1}, \ref{fig:Q1toQ2}, \ref{fig:Q2toQ3}, and \ref{fig:Q3toQ4}
show the common unfoldings of each consecutive pair of polyhedra in
the refolding sequence from the proof of Theorem~\ref{th:dodeca}.

 \begin{figure*}[thb]
   \centering
   \includegraphics[width=0.8\textwidth]{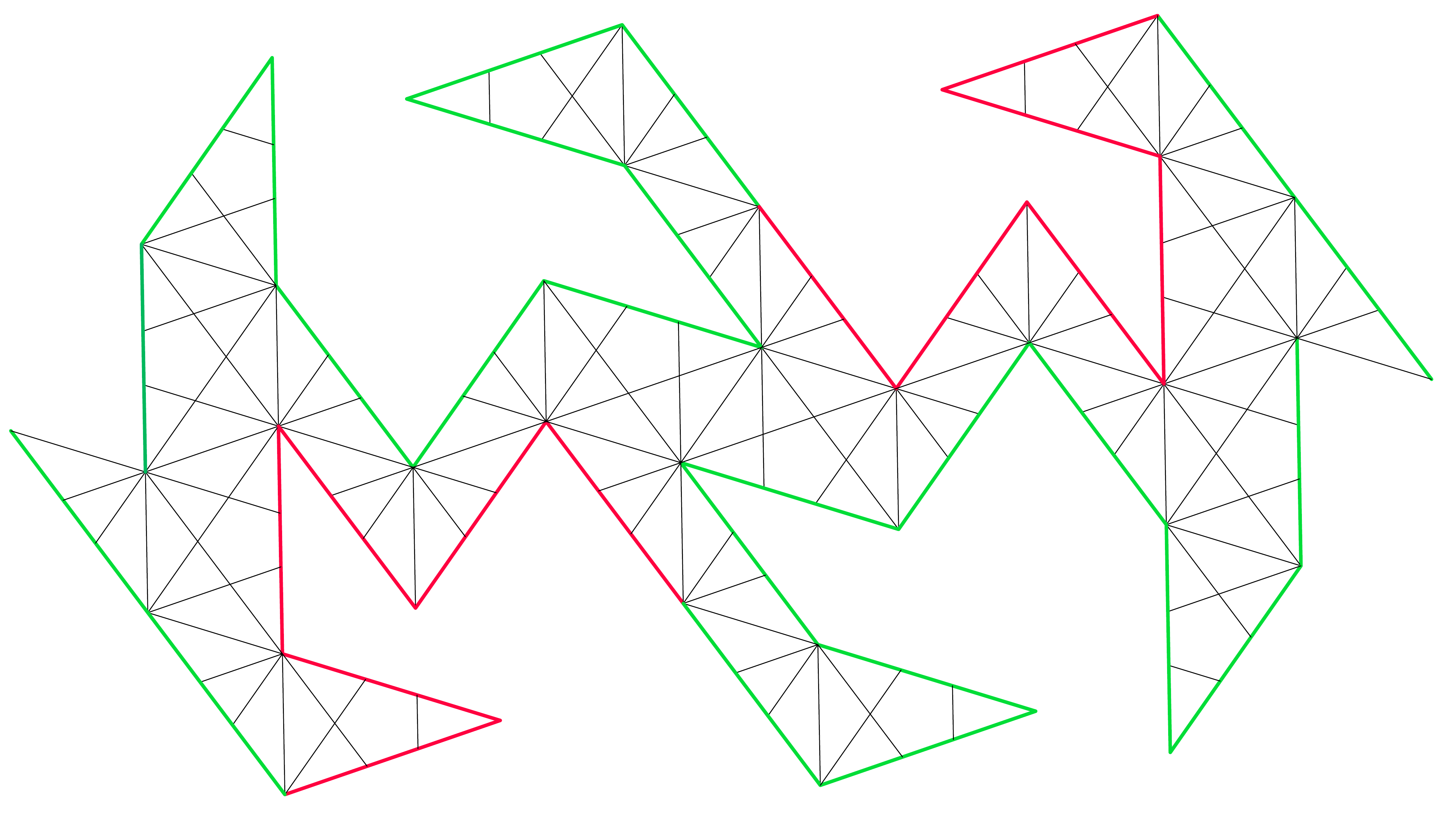}
   \caption{A common unfolding of $D$ and $Q_1$}
   \label{fig:Q0toQ1}
 \end{figure*}

 \begin{figure*}[thb]
   \centering
   \includegraphics[width=0.8\textwidth]{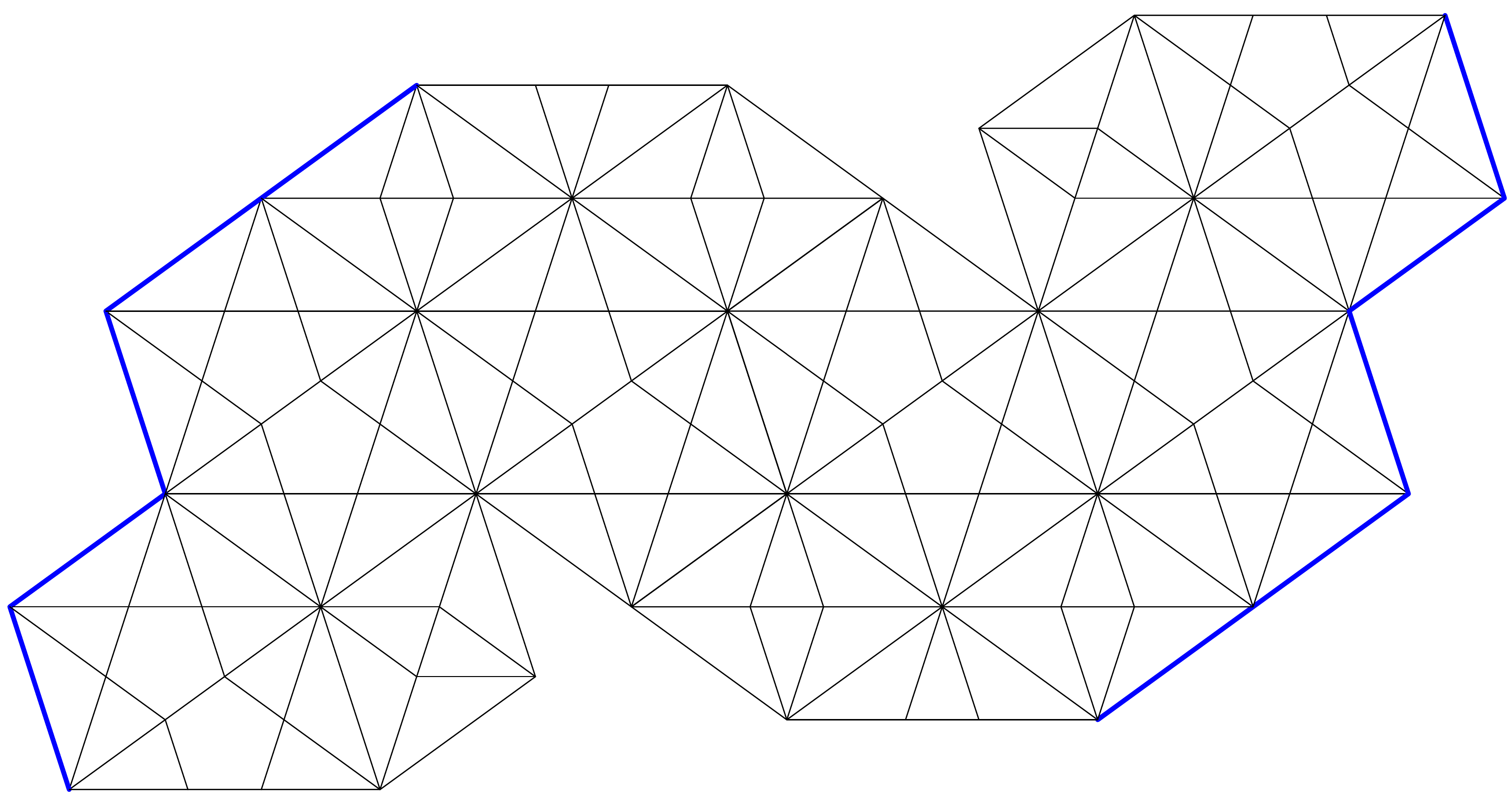}
   \caption{A common unfolding of $Q_1$ and $Q_2$}
   \label{fig:Q1toQ2}
 \end{figure*}

 \begin{figure*}[thb]
   \centering
   \includegraphics[width=0.8\textwidth]{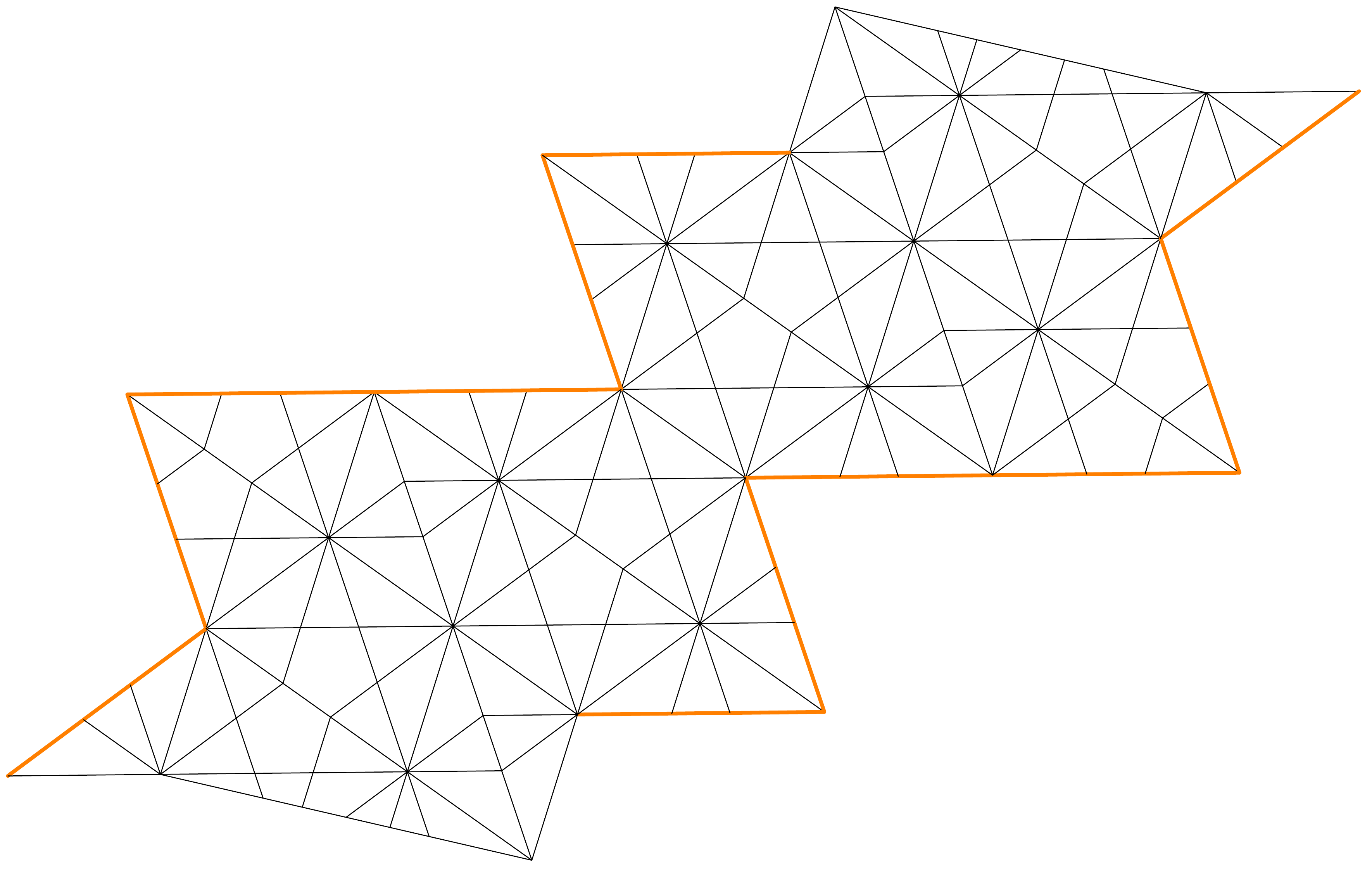}
   \caption{A common unfolding of $Q_2$ and $Q_3$}
   \label{fig:Q2toQ3}
 \end{figure*}

 \begin{figure*}[thb]
   \centering
   \includegraphics[width=0.8\textwidth]{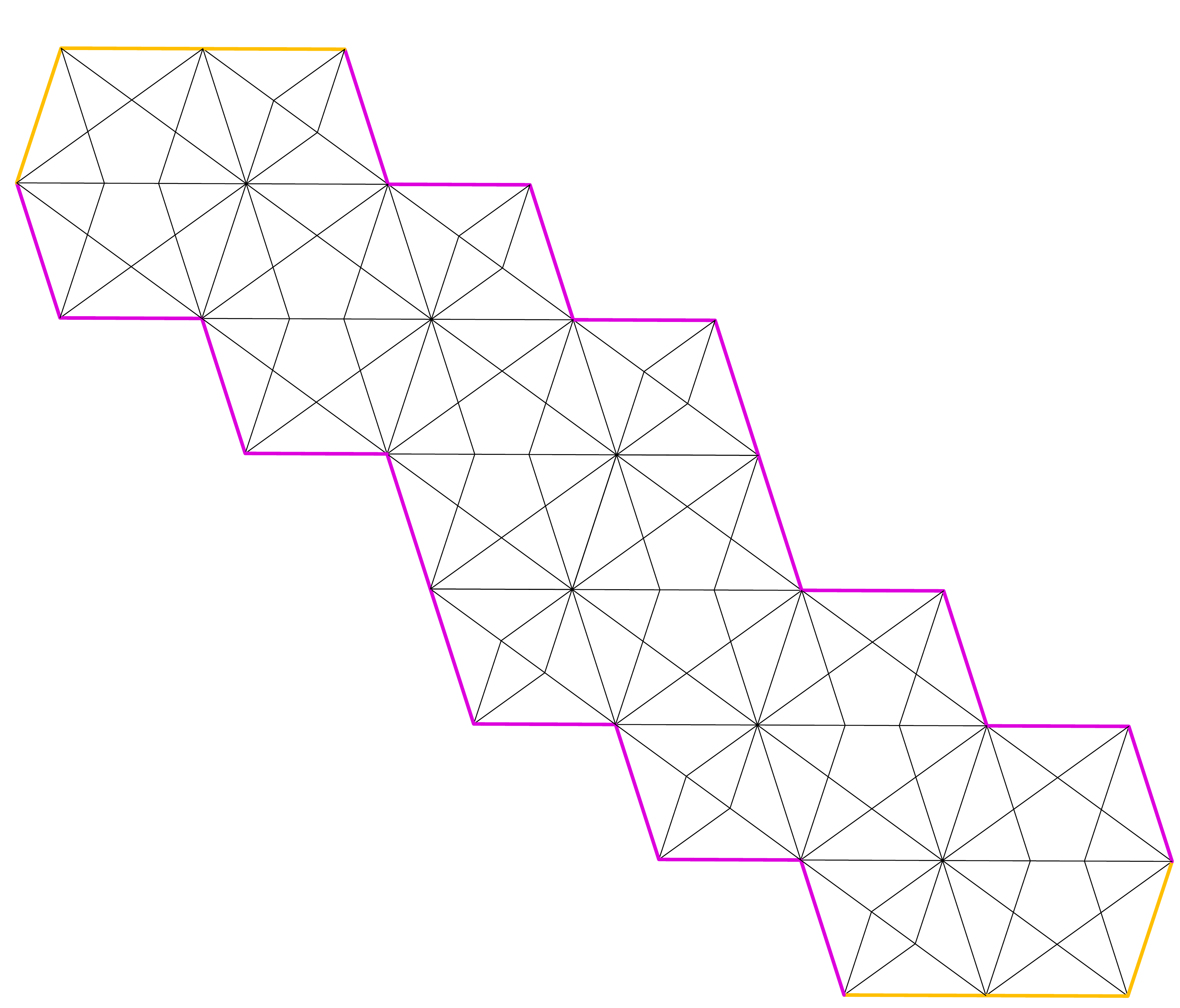}
   \caption{A common unfolding of $Q_3$ and $Q_4$}
   \label{fig:Q3toQ4}
 \end{figure*}

\end{document}